\documentclass[letterpaper,twocolumn,10pt]{article}
\usepackage{usenix-2020-09}

\usepackage{tikz}
\usetikzlibrary{decorations,shapes,backgrounds,calc}
\tikzstyle{msg}=[->,black,>=latex]
\tikzstyle{rubber}=[|<->|]
\tikzstyle{announce}=[draw=blue,fill=blue,shape=diamond,right,minimum
  height=2mm,minimum width=1.6667mm,inner sep=0pt]
\tikzstyle{decide}=[draw=red,fill=red,shape=isosceles triangle,right,minimum
  height=2mm,minimum width=1.6667mm,inner sep=0pt,shape border rotate=90]
\tikzstyle{cast}=[draw=green!50!black,fill=green!50!black,shape=circle,left,minimum
  height=2mm,minimum width=1.6667mm,inner sep=0pt]

\usepackage{alltt}

\usepackage{amsmath}
\usepackage{mathtools}
\usepackage{algorithm2e} 
\usepackage{tabularx}
\usepackage{hyperref}
\usepackage{subcaption}
\usepackage{makecell}
\usepackage{multirow}
\usepackage{graphicx}
\usepackage{epstopdf}

\usepackage{amssymb}
\usepackage{rounddiag}
\usepackage{float}

\usepackage{booktabs}
\usepackage{technote}
\usepackage{homodel}
\usepackage{enumerate}

\usepackage{algorithmicplus}

\usepackage[english]{babel}
\usepackage{amsthm}

\usepackage{filecontents}

\newcommand{\alterbft}{AlterBFT\xspace}
\newcommand{\fastalter}{FastAlterBFT\xspace}
\newcommand{\mymsg}[2]{\ensuremath{\langle #1 \rangle_{#2}}}

\usepackage{enumitem}

\usepackage{caption}
\usepackage{hyperref}

\newconstruct{\FOREACH}{\textbf{for each}}{\textbf{do}}{\ENDFOREACH}{}

\newcommand\Broadcast{\textbf{broadcast}}
\newcommand\Forward{\textbf{forward}}

\newident{noop}
\newconstruct{\UPON}{\textbf{upon}}{\textbf{do}}{\ENDUPON}{}
\newconstruct{\WHEN}{\textbf{when}}{\textbf{do}}{\ENDWHEN}{}

\newconstruct{\FUNCTION}{\textbf{Function}}{\textbf{:}}{\ENDFUNCTION}{}
\newconstruct{\PROCEDURE}{\textbf{Procedure}}{\textbf{:}}{\ENDPROCEDURE}{}

\def\PAR#1{\par\medskip\noindent\textbf{#1}}

\newcommand{\typeS}{$\mathcal{S}$\xspace}
\newcommand{\typeL}{$\mathcal{L}$\xspace}

\newtheorem{remark}{Remark}[section]

\usepackage{booktabs}
\newcommand{\ra}[1]{\renewcommand{\arraystretch}{#1}}

\begin{document}


\date{}

\title{\Large \bf Message Size Matters: \alterbft's Approach to\\ Practical Synchronous BFT in Public Clouds}

\author{
    {\rm Nenad Milo\v{s}evi\'c}\textsuperscript{1} \quad
    {\rm Daniel Cason}\textsuperscript{2} \quad
    {\rm Zarko Milo\v{s}evi\'c}\textsuperscript{2} \quad
    {\rm Robert Soulé}\textsuperscript{3} \quad
    {\rm Fernando Pedone}\textsuperscript{1} \\[6pt] 
    \textsuperscript{1}Universit\`{a} della Svizzera italiana (USI) \quad
    \textsuperscript{2}Informal Systems \quad
    \textsuperscript{3}Yale University
}

\maketitle

\begin{abstract}
  Synchronous consensus protocols offer a significant advantage over
their asynchronous and partially synchronous counterparts by providing
higher fault tolerance---an essential benefit in distributed systems,
like blockchains, where participants may have incentives to
act maliciously. However, despite this advantage, synchronous
protocols are often met with skepticism due to concerns about their
performance, as the latency of synchronous protocols is tightly linked
to a conservative time bound for message delivery.

This paper introduces AlterBFT, a new Byzantine fault-tolerant
consensus protocol. The key idea behind AlterBFT lies in the new model
we propose, called hybrid synchronous system model. The new model is
inspired by empirical observations about network behavior in the public
cloud environment and combines elements from the synchronous and
partially synchronous models. Namely, it distinguishes between small
messages that respect time bounds and large messages that may violate
bounds but are eventually timely. Leveraging this observation,
AlterBFT achieves up to 15$\times$ lower latency than state-of-the-art
synchronous protocols while maintaining similar throughput and the
same fault tolerance. Compared to partially synchronous protocols,
AlterBFT provides higher fault tolerance, higher throughput, and
comparable latency.

  \end{abstract}

\section{Introduction}
\label{sec:introduction}

State machine replication (SMR) \cite{lamport-time-clocks,schneider90}
is a fundamental approach to fault tolerance used in many critical
applications and services.  Most deployed systems assume \emph{crash failures} \cite{burrows06,
  glendenning11,corbett12}, which do not
account for a range of real-world issues, including malicious
behavior, software bugs, bit flips, and more.  However, with the rise
of blockchain technology, interest in protocols that tolerate
arbitrary failures, i.e., \emph{Byzantine fault-tolerant} (BFT)
protocols~\cite{byz-gen}, has grown significantly. In blockchain
systems, participants may be incentivized to act maliciously, and
thus, BFT protocols play a crucial role in maintaining the security
and integrity of the system.

BFT protocols differ in their assumptions about the underlying
network.  Two common characterizations are the synchronous and the
partially synchronous system models.  In synchronous systems, there is
a known \emph{permanent} upper time-bound $\Delta$ for messages to
be sent from one participant to another.  The partially synchronous
system model assumes an \emph{eventual} upper time-bound which holds
after an unknown moment in the execution, referred to as the global
stabilization time (GST) \cite{dwork88}.

\PAR{Motivation.}
With respect to fault tolerance, synchronous protocols have a distinct
advantage over partially synchronous protocols, as they only require
$2f + 1$ replicas to tolerate $f$ malicious players~\cite{fitzi,
  majority-proof-1,majority-proof-2}.  In contrast, because partially
synchronous protocols relax the assumptions about the network, they
need at least $3f + 1$ replicas~\cite{dwork88}.  In a blockchain
system with 100 participants, for example, a synchronous protocol can
handle up to 49 malicious replicas, while a partially synchronous
protocol can only tolerate 33.  Moreover, smaller quorum sizes not
only enhance robustness but also improve performance, as receiving
responses from fewer replicas is typically faster, especially in
wide-area networks (WANs).

Since the safety of synchronous protocols relies on the permanent
upper time bound that must always hold, this bound must be chosen
conservatively, unlike in partially synchronous protocols. A typical
approach in synchronous BFT systems is to measure round-trip time
(RTT) latency beween participants using ping, and choose a value of
$\Delta$ that is either the 99.99th percentile or a multiple of the
maximum observed value~\cite{xft,sync-hotstuff}. This approach is
problematic for at least two reasons. First, the conservative choice
of $\Delta$ results in poor performance. Second, the method used to
determine the value may not reflect precise RTTs for messages because
messages in blockchain systems are significantly larger than pings.\footnote{Messages
carrying blocks can even reach a few megabytes in size:
\href{https://www.blockchain.com/explorer/charts/avg-block-size}{https://www.blockchain.com/explorer/charts/avg-block-size}.}

\PAR{Key Observation.}
To get a more accurate understanding of how message size influences
latency and, in turn, the performance and resilience of synchronous
protocols, we conducted an empirical study over the course of three
months to measure message delays across various configurations,
deployments, and message sizes.  These measurements were taken in both
intra- and inter-region communication, using different machine
configurations and cloud providers, including Amazon Web Services
(AWS) and DigitalOcean. The results consistently revealed a key
phenomenon: in all setups, small messages exhibited significantly lower latency and
less variance than large messages, sometimes by up to
two orders of magnitude.  As a result, synchronous protocols must
account for message size when choosing $\Delta$.  And protocols that
exchange large messages and depend on their timely delivery for safety
must adopt a larger $\Delta$ to account for their greater and more
variable latencies. This, in turn, directly increases the latency of
such protocols, ultimately impacting their performance.

\PAR{New Approach.}
This paper introduces the hybrid synchronous system model,
differentiating small and large messages.  The new model assumes that
for small messages, there is a permanent time bound that always holds,
similar to the synchronous system model; for large messages, there is
an eventual time bound, which holds after the global stabilization
time (GST), as in the partially synchronous model.

Based on this model, we designed a novel BFT consensus protocol,
\alterbft. \alterbft separates messages used for coordination from
messages used for value propagation.  Coordination messages are
small---in our experience, less than 4KB---and value propagation
messages can be of arbitrary size. By making this distinction, \alterbft's
safety (i.e., no two honest replicas decide on a different value) 
only depends on the timely delivery of small
messages; therefore, AlterBFT can improve performance
without limiting value size.  While distinguishing between these
message types might seem like an obvious design choice, it raises
significant complications to ensure that (i) participants can vote for a
value before seeing the value and (ii) participants
who have voted for a value eventually receive it. We explain in the paper
how \alterbft guarantees (i) and (ii), and
present a detailed proof of correctness for AlterBFT in Appendix~\ref{sec:appendix-proof}.

\PAR{Results.} 
We have implemented \alterbft and compared it to state-of-the-art
synchronous and partially synchronous protocols.  Experimental
evaluation in a geographically distributed environment shows that
\alterbft improves the latency of synchronous protocols from
1.5$\times$ to 14.9$\times$, achieving latency comparable to partially
synchronous protocols. Furthermore, \alterbft achieves similar
throughput as synchronous protocol, consistently higher than partially
synchronous protocols, from 1.3$\times$ to 7.2$\times$.  Lastly,
\alterbft tolerates the same number of failures $f<n/2$ as synchronous
protocols, an improvement over partially synchronous protocols, where
$f<n/3$.

\PAR{Roadmap.}
The remainder of the
paper is structured as follows.  Section~\ref{sec:key-observation}
provides more details on the motivation and opportunity.  
Section \ref{sec:sysmodel} details the system model and main assumptions.
Section \ref{sec:problem-def} defines the problem solved by \alterbft rigorously.
Section \ref{sec:alterbft} presents \alterbft.  Section \ref{sec:evaluation}
experimentally evaluates \alterbft's performance in a geographically
distributed environment and compares \alterbft to state-of-the-art
synchronous and partially synchronous protocols.  Section
\ref{sec:related-work} overviews related work and Section
\ref{sec:conclusion} concludes. All appendices are included as
supplementary material.
\section{Motivating Observation}
\label{sec:key-observation}

Synchronous BFT protocols are often viewed with skepticism in the
distributed systems community. The primary reason for this stems from
the challenge of determining a message bound that is both sufficiently
large to guarantee protocol correctness and tight enough to provide
good performance.  For example, Sync HotStuff~\cite{sync-hotstuff}
uses a bound of 50$\times$ the largest observed latency in local-area
setups, while XFT~\cite{xft} employs the 99.99th percentile of
latencies for wide-area setups. 

To investigate opportunities to improve performance, we
conducted a three-month study to understand assumptions on message
bounds made in existing synchronous systems. This assumption is
important because synchronous protocols run at the pace of the
$\Delta$ value, as opposed to partially synchronous systems, which run
at the pace of a quorum of replicas. Our goal was to understand
how message size relates to latency.

We assume a geographically distributed system that relies on message
passing for communication. Specifically, we focus on public cloud
environments. This section presents the most relevant
results, and, to be concise, we do not show measurements for all
message sizes, but only for messages of 2KB and 128KB. Appendix
\ref{sec:appendix-msg-delays} contains the full study.

\begin{table*}[ht]
\centering
\begin{tabular}{lcccccc}
\toprule
            & \multicolumn{3}{c}{\textbf{99.99\%}} & \multicolumn{3}{c}{\textbf{MAX}} \\
\cmidrule(lr){2-4} \cmidrule(lr){5-7}
\textbf{}   & \textbf{2KB} & \textbf{128KB} & \textbf{Diff} & \textbf{2KB} & \textbf{128KB} & \textbf{Diff} \\
\midrule
\textbf{Single Region}	& 5.13   	& 120.48	& 23.49$\times$	& 10.87	& 180.10	& 16.57$\times$ \\
\textbf{Large Machines}	& 1.01   	& 3.99	& 3.94$\times$		& 6.64	& 107.34	& 16.15$\times$ \\
\textbf{Cross-Region}	& 197.50 	& 1399.00	& 7.08$\times$		& 2008.50	& 7295.50	& 3.63$\times$  \\
\textbf{Different Provider}	& 383.00 	& 4953.50	& 12.93$\times$	& 591.50	& 5879.00	& 9.94$\times$  \\
\textbf{Cross-Vendor}	& 1114.00	& 5976.00	& 5.36$\times$		& 4625.50	& 6558.00	& 1.42$\times$  \\ \hline
\textbf{Synthetic model 1}	& 5.13	& 8.13	& 1.59$\times$		& \multicolumn{3}{l}{(based on Single Region)} \\
\textbf{Synthetic model 2}	& 1.01	& 2.11	& 2.01$\times$		& \multicolumn{3}{l}{(based on Large Machines)} \\
\bottomrule
\end{tabular}
\caption{Latency comparison in milliseconds across different setups (99.99\% and MAX).}
\label{tbl:delta-msg-size}
\end{table*}

\PAR{Single-Region Experiments.} 
Our initial
experiments were conducted using two replicas located in the same AWS
region, N. Virginia. These replicas were hosted on free-tier t3.micro
instances. The results, summarized in
Table~\ref{tbl:delta-msg-size}, revealed a surprising pattern: there
was a distinct bifurcation in latency based on message size.
Concretely, large messages of 128KB displayed more than 23$\times$ the latency of
smaller messages of 2KB. Moreover, on Figure \ref{fig:cdf-delay}, we can clearly see that 
the latencies of large messages are not only larger but also vary much more compared to small messages.


  \begin{figure}[ht]
    \center
      \includegraphics[width=\columnwidth]{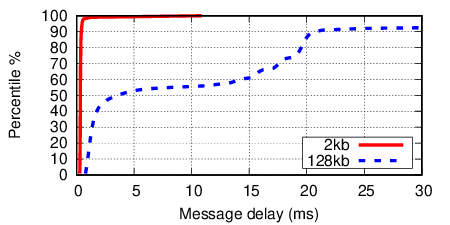}
      \caption{Communication delays between two replicas located in the the same AWS region (N. Virginia).}
      \label{fig:cdf-delay}
    \end{figure}

\PAR{Large-Machine Experiments.} 
Our first intuition was to suspect that the performance was related to
our choice of instance, as the free-tier are listed as having ``low to
moderate'' network capabilities. Perhaps the free-tier instances were
rate limited in some way or configured with a different network
setting? So, we repeated the same experiment using larger instances,
specifically m5.8xlarge machines, which are commonly used in recent
blockchain protocols~\cite{narwhal}. These machines feature 128 GB of
memory, 32 vCPUs (16 physical cores), and a network bandwidth of 10
Gbps. The results indicated that larger machines significantly reduce
message delays. However, despite the overall reduction, there was
still a clear bimodal pattern, with larger messages having almost
4$\times$ the latency of smaller messages.

\PAR{Cross-Region Experiments.} 
We wondered if the behavior was because the experiments ran in a
single region. To test this hypothesis, we again repeated the
experiment, but modified the setup so that five replicas were placed
in the most distant AWS regions—N. Virginia, S. Paulo, Stockholm,
Singapore, and Sydney. Once again, the results supported our initial
observation. Large messages had 7$\times$ the latency of smaller
messages.

\PAR{Different Provider Experiments.} 
One explanation for the observed phenomenon is that perhaps Amazon AWS
sets higher priorities for small messages, or uses some type of unique
configuration.  To check this hypothesis, we repeated the experiment
on a different cloud provider, DigitalOcean. We placed the replicas in
five distinct regions: New York, Toronto, Frankfurt, Singapore, and
Sydney.  Once again, we saw the same behavior: there was a clear
distinction between the latency of small and large messages, with
large messages having 13$\times$ the latency of smaller messages.

\PAR{Cross-Vendor Experiments.} 
Finally, we asked if the same pattern held when sending traffic
between different cloud providers. Maybe there was something
particular to intra-provider communication, as opposed to
inter-provider communication? To eliminate this potential bias, we
performed experiments across providers, placing replicas in the same
five AWS and five DigitalOcean regions and measuring the latencies
between them. Once again, the observation held.

\PAR{Synthetic Model.}
While we cannot completely explain why small messages have lower
latency and less jitter than large messages, one fundamental factor is
the way networks handle message fragmentation.  When a large message
(e.g., 128KB) is transmitted, it is divided into multiple packets at
the network level. Each packet typically corresponds to a size defined
by the Maximum Transmission Unit (MTU). These fragmented packets are
then sent individually through the network, and on the receiving side,
all packets must arrive before the message can be reconstructed. If
any of these packets are delayed, the original message will be
delayed.

To test this hypothesis, we developed a synthetic model that simulates
the delay of 128KB messages when fragmented into 2KB messages.  We
then used the delay data collected from our small message experiments
(2KB messages) and simulated the delay for 128KB messages.  To
calculate the delay of a large message, we randomly selected 64
packets from the small message dataset.  The largest delay among the
sample packets represents the delay of the 128KB message.

The synthetic model is a theoretical best case as it ignores aspects
of large message transmission (e.g., network congestion, packet loss
and retransmission, packet reorder).  However, it provides insight
into the difference between the delay of small and large messages.

\PAR{Salient Observations.} 
From these results, we make the following observations:
\begin{itemize}[leftmargin=*]
\item By examining a range of message sizes, we empirically saw that messages up to 4KB
  demonstrated significantly lower latency and more stability than larger messages (see Appendix \ref{sec:appendix-msg-delays}).
\item This suggest a model 
that distinguishes between small ($\leq$4KB) and large ($\geq$4KB) message sizes.
\item For small messages, we can
set tight bounds on the message delivery time, $\Delta$, like synchronous systems.  
\item For large messages, we can treat the system like a partially synchronous system and
incorporate a global stabilization time (GST).   
\end{itemize}  

The challenge, then, is to design a protocol that can benefit from this model:
a protocol whose safety relies on small messages only, 
and thus requires only a simple majority of honest replicas, 
but whose performance can benefit from exchanging large messages.

\section{System Model}
\label{sec:sysmodel}

In this section, we introduce a new system model, called the
\emph{hybrid synchronous} system model.  We start by outlining the
general assumptions and then present the novel timing assumptions,
which are designed to capture the behavior observed in
Section~\ref{sec:key-observation}.

\PAR{General Assumptions.}
We focus on public cloud environments in which geographically
distributed replicas communicate by exchanging messages.  Replicas can
be \emph{honest} or \emph{faulty}.  An honest replica follows its
specification; a faulty, or Byzantine, replica presents arbitrary
behavior.  The system includes $n$ replicas, among which up to $f$ may
be faulty, with the condition that $n > 2f$.  Replicas do not have
access to a shared memory or a global clock, but each replica has its
own local (hardware) clock, and while these clocks are not
synchronized, they all run at the same speed.  Replicas communicate
using point-to-point reliable links: if both sender and receiver are
honest, then every message sent is eventually received.

We assume the presence of a public-key infrastructure (PKI), secure
digital signatures, and collision-resistant hash functions.  A message
$m$ sent by process $p$ is signed with $p$'s private key and denoted
as $\li{m}_p$.  Additionally, $id(v)$ represents the invocation of a
random oracle that returns the unique hash of value $v$.


\PAR{Timing Assumptions.}\label{sec:timeassump}
Driven by the experimental data presented in Section
\ref{sec:key-observation}, we adopt distinct assumptions for small and
large messages. We define ``small'' based on empirical observation as
$\leq 4KB$ (see Appendix \ref{sec:appendix-msg-delays}). We assume that small messages
adhere to a predefined time bound, as in the synchronous system
model. In contrast, for large messages, we assume the existence
of an eventual time bound, referred to as the Global
Stabilization Time (GST)~\cite{dwork88}. This assumption
is made by all partially synchrounous consensus protocols~\cite{pbft, bucham18,hotstuff,hotstuff-2,icc,bullshark-partially-sync}.

Thus, our new \emph{hybrid synchronous} system model has two communication properties, one for each message type:

\begin{itemize}[leftmargin=*]
  \item \emph{Type \typeS messages:} If an honest replica $p$ sends a message $m$ of type \typeS to an honest replica $q$ at time $t$, then $q$ will receive $m$ at time $t + \Delta_S$ or before.
  \label{prop-small-msg}
  \item \emph{Type \typeL messages:} If an honest replica $p$ sends a message $m$ of type \typeL to an honest replica $q$ at time $t$, then $q$ will receive $m$ at time $max\{t, GST \} + \Delta_L$ or before. 
  \label{prop-big-msg}
\end{itemize}

Lastly, following the approach of
\cite{sync-hotstuff,rot-sync-hotstuff}, we do not assume lock-step
execution (e.g., \cite{byz-gen,dolev-strong}), where all honest
replicas begin each round (or epoch) simultaneously.  Instead, we
assume that all honest replicas start the execution within a
$\Delta_S$ time.\footnote{This can be implemented in a real system by
having each replica broadcast a small $Start$ message upon the
beginning of the execution. A replica starts either upon receiving a
$Start$ message from another replica or at a specific point in time.}

\PAR{Threat Model.}
We assume that malicious participants can alter their own behaviors
(e.g., can delay sending values, send the wrong value), but they
cannot alter the behavior and communication of honest nodes.
Moreover, they cannot subvert cryptographic primitives. These
assumptions are consitent with prior work~\cite{xft,sync-hotstuff,rot-sync-hotstuff,dfinity}.

\begin{remark}
  Our new system model assumes a majority of replicas in the system
  are honest, while the remaining replicas may behave arbitrarily:
  they can be slow, crash, or act maliciously. Additionally, the model
  assumes that type \typeS messages exchanged between honest replicas
  always adhere to the specified time bound $\Delta_S$, whereas type
  \typeL messages are required to respect the time bound only after
  the Global Stabilization Time (GST).
\end{remark}

\section{Problem Definition}
\label{sec:problem-def}

SMR and consensus are used to totally order client transactions so that replicas process them in 
the same order and remain consistent. Specifically, in blockchain systems, transactions are grouped 
into blocks, and replicas use consensus to agree on a chain of blocks, where the position of a block 
in the chain is referred to as its \emph{height}.

A block $B_k$ at height $k$ has the following format: $B_k \coloneqq (b_k, H(B_{k-1}))$, where $b_k$ 
represents a proposed value (i.e., a set of transactions), and $H(B_{k-1})$ is the hash digest of the 
preceding block. The first block, $B_1 = (b_1, \bot)$, has no predecessor. Every subsequent block $B_k$ 
must specify its predecessor block $B_{k-1}$ by including a hash of it. If block $B_k$ is an ancestor 
of block $B_l$ (i.e., $l \ge k$), we say that $B_l$ \emph{extends} $B_k$.

A valid (blockchain) consensus protocol must satisfy the following properties: 
\begin{itemize}[leftmargin=*]
    \item \emph{Safety}: No two honest replicas commit different blocks at the same height. 
    \item \emph{Liveness}: All honest replicas continue to commit new blocks. 
    \item \emph{External validity}: Every committed block satisfies a predefined 
    \emph{valid()} predicate. 
\end{itemize}

Safety ensures that all honest replicas agree on the same chain of blocks, while liveness guarantees 
that new blocks are continuously added to the blockchain by honest replicas, preventing the system 
from halting. External validity is, as the name suggests, an application-specific property, 
defined to ensure certain conditions are met (e.g., the absence of double-spending transactions).

\begin{remark}
    \alterbft is the first (blockchain) consensus protocol designed for the hybrid 
    synchronous system model. Its safety mechanism depends solely on type \typeS messages 
    exchanged between honest replicas, enabling it to tolerate the same number of faulty replicas 
    as synchronous protocols. Large type \typeL messages are used to carry blocks, facilitating high 
    performance by enabling agreement on substantial amounts of data. However, their timely delivery 
    is necessary only for liveness, similar to partially synchronous protocols.
\end{remark}

In this work, we focus on leader-based consensus protocols 
\cite{lamport98, pbft, bitcoin, ethereum, avalanche}, where a designated replica, the leader, proposes the 
order of blocks. Specifically, we examine rotating-leader consensus protocols 
\cite{rot-sync-hotstuff, hotstuff-2, hotstuff, bullshark-partially-sync, bucham18, fireledger, streamlet, internet-computer}, 
where leadership changes regularly, not only when the leader is faulty (e.g., \cite{lamport98, xft, pbft, sync-hotstuff}).

Rotating leadership enhances censorship resistance and fairness in blockchain systems by giving every replica an opportunity to propose blocks and earn rewards. Frequent leader changes also mitigate the impact of Byzantine leaders, limiting the ability of Byzantine leaders to censor transactions or exploit the control of block creation (MEV \cite{mev}).


\newcommand\Disseminate{\textbf{Disseminate}}

\newcommand\Proposal{\mathsf{PROPOSAL}}
\newcommand\ProposalPart{\mathsf{PROPOSAL\mbox{-}PART}}
\newcommand\PrePrepare{\mathsf{INIT}} \newcommand\Prevote{\mathsf{PREVOTE}}
\newcommand\Precommit{\mathsf{PRECOMMIT}}
\newcommand\Decision{\mathsf{DECISION}}

\newcommand\ViewChange{\mathsf{VC}}
\newcommand\ViewChangeAck{\mathsf{VC\mbox{-}ACK}}
\newcommand\NewPrePrepare{\mathsf{VC\mbox{-}INIT}}
\newcommand\leader{\mathsf{leader}}

\newcommand\newHeight{newHeight} \newcommand\newRound{newRound}
\newcommand\nil{nil} \newcommand\id{id}
\newcommand\prevote{prevote} \newcommand\prevoteWait{prevoteWait}
\newcommand\precommit{precommit} \newcommand\precommitWait{precommitWait}
\newcommand\commit{commit}
\newcommand{\propose}{\textsc{propose}}
\newcommand{\vote}{\textsc{vote}}
\newcommand{\silence}{\textsc{silence}}
\newcommand{\double}{\textsc{equiv}}
\newcommand{\quit}{\textsc{quit-epoch}}
\newcommand{\silCert}{\textsc{silence-cert}}
\newcommand{\valueCert}{\textsc{block-cert}}
\newcommand{\equivCert}{\textsc{equiv-cert}}
\newcommand{\ready}{\textsc{active}}
\newcommand{\committed}{\textsc{committed}}
\newcommand{\notcommitted}{\textsc{not-committed}}

\newcommand\timeoutCertificate{timeoutCertificate}
\newcommand\timeoutEpochChange{timeoutEpochChange}
\newcommand\timeoutCommit{timeoutCommit}
\newcommand\timeoutPrecommit{timeoutPrecommit}
\newcommand\proofOfLocking{proof\mbox{-}of\mbox{-}locking}

\section{\alterbft}
\label{sec:alterbft}

This section presents an overview of \alterbft, describes how
\alterbft\ operates when the leader is honest and when it is
Byzantine, argues about \alterbft's correctness, and discusses what happens when small messages violate synchrony bounds.
Appendix \ref{sec:pseudo-code} contains the detailed pseudocode and explanatory comments.

\subsection{Protocol Overview} 
\alterbft builds on the rotating-leader variant of Sync HotStuff \cite{rot-sync-hotstuff}, a state-of-the-art synchronous consensus protocol \cite{sync-hotstuff}.  
\alterbft operates in epochs, numbered $0, 1, 2, \ldots$. In each epoch, a single replica acts as the leader, selected either through a deterministic function or a random oracle.  
The leader's responsibility is to propose a new block to be appended to the blockchain.


\paragraph{Certificates.}  
Certificates are the key abstraction in the protocol. They consist of a set of signed messages that enable a 
replica to prove to itself and other replicas that the leader of an epoch performed specific actions. 
There are three types of certificates:  
\begin{itemize}[leftmargin=*]
    \item \textbf{Block certificate:} Verifies that the leader of an epoch proposed a valid block.  
    \item \textbf{Equivocation certificate:} Proves that the leader is Byzantine and proposed multiple 
    conflicting blocks in the same epoch.  
    \item \textbf{Silence certificate:} Demonstrates that the leader failed to send a proposal to at least 
    one honest replica, either due to being slow or remaining silent.  
\end{itemize}

Certificates are used in Sync HotStuff to ensure that all replicas 
enter each epoch within $\Delta$ time and to detect whether it is safe to commit the proposed block in an 
epoch. These goals are achieved through a simple exchange of certificates. Specifically:  
\begin{itemize}[leftmargin=*]
    \item When a replica wants to move to the next epoch after receiving a certain certificate, it can 
    simply broadcast the certificate, ensuring that all honest replicas will receive it within $\Delta$ time. 
    Since the certificate is self-contained, any replica that receives it can verify its validity and also 
    proceed to the next epoch.  
    \item When a replica wants to commit the proposed block in epoch $e$, it sends a block certificate to 
    all other replicas and sets a $2\Delta$ timeout. Within $\Delta$ time, all honest replicas will receive 
    the certificate. If any honest replica possesses a certificate indicating that the leader is faulty 
    (either an equivocation or a silence certificate), it can forward this certificate to others. 
    Within $\Delta$ time, the faulty leader's behavior will be known, and replicas will avoid committing 
    the block.  
\end{itemize}

We aimed to use the same mechanism in \alterbft; however, since the safety of \alterbft must rely exclusively 
on small messages, certificates must be categorized as type \typeS messages. 
Unfortunately, the certificates in Sync HotStuff depend on large messages and therefore require modifications. 
Specifically:  
\begin{itemize}[leftmargin=*]
    \item The equivocation certificate consists of two signed proposals from the leader, each containing full 
    conflicting blocks.  
    \item The block certificate requires a replica to receive the original proposal, including the block.  
\end{itemize}

Since blocks can be of arbitrary size, messages carrying them must be categorized as type \typeL. Consequently, 
before GST, such messages might be delayed, preventing honest replicas from entering the epoch within $\Delta$ 
time or potentially compromising safety.

\paragraph{Detecting equivocation without the proposal.}  
To detect equivocation without relying on the full block proposal, we made a simple but effective observation: 
detecting an equivocating leader requires only proof that it sent signed votes for two different proposals. 
Consequently, in \alterbft, the leader must send a signed vote alongside its proposal, which contains only 
the hash of the block. A replica considers a proposal valid only if it receives both the proposal and 
the accompanying vote. With this modification, an equivocation certificate now contains only  
two signed votes from the leader for two different hashes, making it small and suitable for categorization as a type \typeS message.

\paragraph{Ensuring full block availability.}  
The second issue concerned the requirement that a replica must receive a full block proposal, 
along with $f+1$ votes, to move to the next epoch and consider the block certified. 
We observed that this is unnecessary. Instead, a replica can progress after receiving $f+1$ votes alone, 
as it knows that at least one of these votes comes from an honest replica that has received the full proposal. 
This honest replica will then forward the proposal, ensuring that all honest replicas eventually receive it.  

\begin{remark}
All certificates in \alterbft are small and categorized as type \typeS messages. 
This ensures that replicas in \alterbft can enter each epoch within $\Delta_S$ time. Additionally, once a 
block is certified, replicas can safely commit it after $2\Delta_S$ time.
\end{remark}

\subsection{Epoch with an Honest Leader}

\alterbft shares the communication pattern of Sync HotStuff \cite{rot-sync-hotstuff} when the leader is honest. 
This pattern is proven to achieve optimal latency in a rotating-leader setup and supports responsive leader 
rotations under a sequence of honest leaders---a property known as optimistically responsive leader rotations.
The protocol relies on certificates, and consists of three phases: propose, vote, and commit.

%

\PAR{Certified Blocks.}
Before a block $B_k$ can be added to the blockchain, it must be
certified.  Specifically, this means it must be approved or voted for
by at least one honest replica.  Since our system tolerates up to $f$
Byzantine replicas, a block is considered certified when it receives
$f+1$ signed votes from distinct replicas within an epoch. We denote
the certificate for block $B_k$ from epoch $e$ as $C_e(B_k)$.

In \alterbft, each honest replica keeps track of the most recent
certified block it knows of in a variable called $lockedBC$.  A
certificate $C_e(B_k)$ is considered more recent than $C_{e'}(B'_k)$
if $e>e'$. Honest replicas use this value to guard the safety of
\alterbft.

\PAR{Proposal Phase.} 
In an epoch, the leader, an honest replica $l$, broadcasts a proposal
containing a block that extends its most recent certified block,
$lockedBC_l$. Before proposing a new block, the honest leader must
ensure it has the most recent certified block. If the leader already
possesses a block certified in the previous epoch, it knows this is
the most recent certified block and can propose a new block
immediately. Otherwise, it must first learn the certificates from
other honest replicas.

To achieve this, the leader uses a timeout called
$timeoutEpochChange$. Specifically, the leader starts this timeout and
sets it to expire in $2\Delta_S$. Since the leader knows that within
$\Delta_S$ all honest replicas will enter the same epoch and that
every honest replica broadcasts its $lockedBC$ before entering the new
epoch, the leader will receive the required information within
$2\Delta_S$, before $timeoutEpochChange$ expires.  Note that the
leader waits for this timeout only if the previous epoch was before
GST or if the leader in the previous epoch was Byzantine.

The proposal message includes the new block , $lockedBC$, and the
epoch number. It is classified as a type \typeL message because it
carries the block, which may contain an arbitrary number of
transactions.  In addition to the proposal, the leader also broadcasts
its signed vote for the proposal.  This vote contains the current
epoch number and the hash of the proposed value, $id(v)$. Since the
vote is of constant small size, it is considered a type \typeS
message. \alterbft uses the leader's vote message to detect
equivocation rather than relying on the proposal message itself.

\paragraph{Vote Phase.} 
Upon receiving a proposal and a signed vote from the leader, a replica 
verifies whether the block is valid (Section \ref{sec:problem-def}). 
Furthermore,  the replica will accept the new block only if it truly 
extends the block from $lockedBC_l$ and if $lockedBC_l$ is at least 
as recent as the replica's $lockedBC$. The replica votes for a proposal 
by sending a signed vote message to all replicas, voting only once per 
epoch for the first proposal it receives. Additionally, the replica 
forwards the proposal and the leader's vote to all replicas. 

Forwarding the leader's vote is necessary to detect leader equivocation.
Forwarding the proposal ensures the eventual delivery of all certified 
blocks. If a block is certified, at least one replica among those who 
voted for the value is honest and will forward the block.\footnote{Forwarding the proposal can be disabled, in which case a pull 
mechanism would need to be implemented to download the missing blocks \cite{mysticeti}.}

\paragraph{Commit Phase.}
Upon receiving a block certificate for the block proposed in the current epoch, 
an honest replica updates its $lockedBC$. It then propagates the block certificate 
to all replicas and starts a $2\Delta_S$ timer, $timeoutCommit$. These actions 
are performed even if the replica has not yet received the full proposal message, 
as it knows the proposal will eventually arrive. This is ensured because at 
least one of the replicas that received the full block and voted for it is
honest and will forward it.

When $timeoutCommit$ expires, if the replica has not received any other certificate 
(e.g., an equivocation certificate, a silence certificate, or a block certificate for 
a different block), it commits the proposed block. If the full block has not yet arrived, 
the replica stores the hash and commits the block once it is received.  
The replica is safe to commit this block because it knows that all honest replicas 
received the block certificate as the first certificate in this epoch and updated 
their $lockedBC$. Moreover, it ensures that all honest replicas voted for this block
and as a result this is the only possible block certificate for the epoch. 
Consequently, only blocks extending this block can be certified and committed in subsequent 
epochs (see Section \ref{sec:alter-safety}). 

Since blocks are linked, committing block $B_k$ also commits all blocks it extends. 
Specifically, when a replica commits block $B_k$ proposed in the current epoch, 
it directly commits $B_k$ and indirectly commits all its ancestors.

\begin{remark}
   \alterbft achieves an optimal latency of $2\Delta$ in a rotating-leader setup, 
   similar to Sync HotStuff. However, in \alterbft, $\Delta$ accounts only for 
   small messages, denoted as $\Delta_S$, which is significantly smaller, resulting 
   in a lower latency.
\end{remark}

Lastly, a replica starts the next epoch as soon as it receives a block certificate, 
without waiting for $timeoutCommit$ to expire. This is safe because the replica 
knows it is the most recent possible block certificate. In the following epoch, 
if the replica becomes the leader, it will extend this block without waiting 
for $timeoutEpochChange$. All honest replicas will accept the new block, 
as it extends the block certified in the previous epoch, representing the 
most recent block certificate. 

\begin{remark}
   \alterbft achieves optimistically responsive leader rotations. 
   During a sequence of honest leaders, replicas progress through epochs responsively, 
   requiring only a block certificate in each epoch to move forward.
\end{remark}

\paragraph{FastAlterBFT.} 

Although \alterbft achieves optimal latency in a rotating-leader scenario, it is unfortunate that replicas 
must wait for a conservative $2\Delta_S$ before committing a block, even in scenarios where the network 
is synchronous and all replicas in the system are honest---conditions that are the most common in practice.  

To address this limitation, we introduced a fast commit rule to \alterbft \cite{itai-optimistic, sbft, zyzzyva, banyan}.
Specifically, an honest replica commits a block proposed in an epoch if it receives votes for the block from all 
replicas in the system, provided no evidence of misbehavior (e.g., an equivocation or silence certificate) is 
detected. As a result, when there are no failures in the system, \alterbft commits a block in 
just two communication steps, achieving optimal latency \cite{bosco}, 
without waiting for the synchrony bound 
$\Delta_S$. We refer to this fast path as \fastalter.  

It is important to note that adding the fast commit rule does not affect the normal execution of the protocol. 
As described previously, replicas still start the $timeoutCommit$ timer upon receiving a block certificate. 
However, if a replica receives votes from all replicas before the timer expires and no misbehavior is detected, 
it commits the block immediately.

\begin{remark} 
   Without malicious replicas and after GST, \alterbft is fully responsive: it changes leaders and commits blocks at network speed without relying on a conservative $\Delta$.
\end{remark}


\subsection{Epoch with a Faulty Leader}

In epochs with a faulty leader, an honest replica can generate equivocation and silence certificates, in addition to a block certificate.

\begin{itemize}[leftmargin=*]
   \item \textbf{Equivocation certificate ($C_e(\double)$):} Composed of votes signed by the leader for two different blocks within the same epoch.
   \item \textbf{Silence certificate ($C_e(\silence)$):} Composed of $f+1$ silence messages, each one signed by a distinct replica.
\end{itemize}

\alterbft must ensure that Byzantine replicas cannot halt the protocol in epochs with 
a Byzantine leader (maintaining liveness). Additionally, it must guarantee that if an 
honest replica commits a block, all honest replicas update their $lockedBC$ 
to this block before transitioning to the next epoch (maintaining safety).

\PAR{On the need of certificates.} 
%
To prevent an honest replica from remaining stuck in an epoch, \alterbft ensures 
that at least one certificate is created in each epoch. To achieve this, an 
honest replica starts a $timeoutCertificate$ timer upon entering the epoch. 
If the timeout expires and no certificate has been received, the replica sends 
a silence message. If no honest replica has received any other certificate, 
all honest replicas will broadcast silence messages, resulting in the 
creation of $f+1$ silence messages and, subsequently, a silence certificate.  

The duration of $timeoutCertificate$ is $\Delta_L + 4\Delta_S$. 
This duration accounts for the time needed to generate a block 
certificate in epochs with an honest leader after GST. It ensures that replicas 
do not prematurely send silence messages, which could disrupt liveness when the 
leader is honest, and the network is operating after GST.

\PAR{Sometimes waiting is necessary for safety.} 
Similar to the epoch with an honest leader, 
an honest replica starts $timeoutCommit$ upon collecting one of the certificate. 
If the certificate is a block certificate, the replica moves to the next epoch 
immediately without waiting for the timeout to expire. However, if it is a blame or 
equivocation certificate, the replica must wait for the timeout to expire 
or receive a block certificate, before transitioning to the next epoch.

This requirement arises due to the \fastalter commit rule. 
Specifically, there could be a scenario where one honest replica 
commits a block using the fast commit rule, while another replica 
receives a blame or equivocation certificate. If the latter 
replica were allowed to move to the next epoch immediately, 
it would fail to update its $lockedBC$ to the committed block certificate, 
potentially compromising the protocol's safety.

\begin{remark}
   \alterbft's epoch change mechanism improves upon Sync HotStuff's by enabling 
   a responsive commit rule (\fastalter) and enhancing handling of silent or slow 
   leaders. Specifically, \alterbft's $timeoutCertificate$ is set to $\Delta_L + 4\Delta_S$, 
   compared to the equivalent timeout in Sync HotStuff, which is $7\Delta$.
\end{remark}

\subsection{\alterbft's Correctness }
\label{sec:alter-proof-intuition}

In this section, we provide the intuition on \alterbft's correctness. 
The full proof can be found in the Appendix \ref{sec:appendix-proof}.

\subsubsection{Safety}
\label{sec:alter-safety}


To ensure safety (see Section~\ref{sec:problem-def}), \alterbft's commit rules must satisfy two key invariants: if an honest replica $r$ 
commits block $B_k$ in epoch $e$, then (1) $C_e(B_k)$ is the only block certificate that 
exists in epoch $e$ (i.e., no honest replica voted for a block $B'_{k'} \neq B_k$ in $e$), and
(2) all honest replicas lock on $B_k$ by setting $lockedBC$ to $C_e(B_k)$ in epoch $e$.

As a result, in subsequent epochs, honest replicas only vote for blocks that extend those certified in epochs 
$e' \geq e$. Since by (1), $B_k$ is the only certified block in epoch $e$, and
by (2), all honest replicas set $lockedBC$ to $C_e(B_k)$,
in epochs after $e$, honest replicas will only vote for blocks extending $B_k$. 
Consequently, only blocks extending $B_k$ will be certified and committed.

\PAR{How does \alterbft's regular commit rule ensure safety?} 
The regular commit rule states that a replica $r$ commits block $B_k$ if $timeoutCommit(e) = 2\Delta_S$ expires 
and no misbehavior is detected. Invariant (1) holds in this case because $r$, upon receiving a block certificate $C_e(B_k)$ 
at time $t$, starts $timeoutCommit(e)$ and broadcasts $C_e(B_k)$. Since a message containing 
$C_e(B_k)$ is a type \typeS message, all honest replicas will receive it within $\Delta_S$ time, 
by $t+\Delta_S$. 

If any honest replica $q$ voted for $B'_{k'} \neq B_k$, it must have done so before $t + \Delta_S$. 
As $q$ also forwards the leader's vote for $B'_{k'}$, $r$ receives it before $t + 2\Delta_S$. 
In this way, replica $r$ receives the leader's votes for both $B_k$ and $B'_{k'}$ before $timeoutCommit(e)$ expires. 
Consequently, $r$ forms an equivocation certificate $C_e(\double)$ and does not commit. 

Similarly, invariant (2) holds because $q$ will not lock on $C_e(B_k)$ only if it has moved to epoch $e+1$ 
after receiving either $C_e(\double)$ or $C_e(\silence)$ before $t + \Delta_S$. 
Since $q$ forwards the received certificate, and messages carrying certificates are also type \typeS, 
$r$ will receive it before $timeoutCommit(e)$ expires and will not commit.


\PAR{How does \alterbft's fast commit rule ensure safety?} 
In the fast path, a replica commits block $B_k$ in epoch $e$ if it receives votes for $B_k$ from 
all replicas before detecting any misbehavior.  

Invariant (1) trivially holds because a replica knows that all honest replicas voted for 
$B_k$, as it received votes from all honest replicas, and each honest replica votes only once. 

Ensuring invariant (2) prevents an honest replica from always progressing to the next epoch 
immediately after receiving any certificate. Specifically, a replica $r$ locks on block 
$B_k$ at time $t$ and commits at time $t'$, where $t < t' < t + 2\Delta_S$. 
As a result, if an honest replica $q$ receives $C_e(\double)$ or $C_e(\silence)$ at time $t''$, 
where $t' - \Delta_S \leq t'' < t + \Delta_S$, 
it might move to epoch $e+1$ without locking on $C_e(B_k)$, and replica $r$ would not be aware of this. 

To handle this scenario, an honest replica $q$ sets $timeoutCommit(e)$ to expire in $2\Delta_S$ 
after receiving $C_e(\double)$ or $C_e(\silence)$. Moreover, replica $q$ will only move to the next epoch if it receives $C_e(B_k)$ 
or if $timeoutCommit(e)$ expires. 
Since $q$'s $timeoutCommit(e)$ expires at $t'' + 2\Delta_S$, and $t'' + 2\Delta_S > t + \Delta_S$, 
$q$ will receive $C_e(B_k)$ and lock on it before progressing to the next epoch. This ensures that invariant (2) is upheld.

\subsubsection{Liveness}
\label{sec:liveness}

\alterbft guarantees progress after $GST$ (Section~\ref{sec:sysmodel}) during the first epoch led by an honest leader. 
Specifically, progress is ensured in epoch $e > GST$ under an honest leader if: 
(1) the leader proposes a block that all honest replicas accept and vote for, and 
(2) no honest replica broadcasts a \silence\ message in epoch $e$. 
Condition (1) ensures that a block certificate is formed and $timeoutCommit(e)$ is started, while 
condition (2) ensures that no silence certificate can be created. Furthermore, since an honest leader proposes only a single block, 
no equivocation certificate $C_e(\double)$ is possible. 
As a result, when $timeoutCommit(e)$ expires, all honest replicas will commit the proposed block.

To ensure condition (1), the honest leader must learn the most recent certified block before proposing. 
Upon entering epoch $e$ at time $t$, if the honest leader $l$ does not possess a block certificate 
from the previous epoch, $e-1$, it starts $timeoutEpochChange(e)$. 
Since all honest replicas enter epoch $e$ by time $t + \Delta_S$, they broadcast their $lockedBC$ 
no later than $t + \Delta_S$. As a result, the leader $l$ will receive these certificates 
by time $t + 2\Delta_S$. Therefore, the honest leader must set $timeoutEpochChange(e)$ to 
$2\Delta_S$ to ensure it learns the most recent certified block.

To guarantee condition (2), honest replicas must receive the block certificate before $timeoutCertificate(e)$ expires. 
If an honest replica $r$ starts epoch $e$ at time $t$, the honest leader $l$ will enter epoch $e$ 
no later than $t + \Delta_S$. The leader may wait for $timeoutEpochChange(e) = 2\Delta_S$ before proposing a block, 
and thus will propose a block by time $t + \Delta_S + 2\Delta_S$. 
Since the proposal is a message of type \typeL, it may take up to $\Delta_L$ to reach all honest replicas. 
As a result, all honest replicas will vote for the block by time $t + 3\Delta_S + \Delta_L$. Finally, 
as votes are type \typeS messages, an additional $\Delta_S$ is required to deliver them. 
In summary, to guarantee condition (2), honest replicas must set $timeoutCertificate(e)$ to $4\Delta_S + \Delta_L$.





\subsection{Violations of Synchrony}
\label{sec:violations}

\alterbft relies on the timely delivery of small messages
within $\Delta_S$ bound. A natural question is: what happens if this bound is violated? 
Recent work has shown that modern synchronous consensus protocols can tolerate synchrony violations 
without compromising correctness, thanks to their communication diversity and redundancy \cite{opodis}. Specifically, an honest 
replica receives the same information through multiple independent communication paths. As a result, safety violations are 
observed only when many messages violate bound $\Delta$ and Byzantine replicas collude.
As a modern synchronous protocol, these findings also apply to \alterbft. 
Moreover, \alterbft is particularly robust because it requires only type \typeS messages 
to be delivered on time, unlike traditional synchronous protocols that rely on the timely delivery of all messages. 
As discussed in Section \ref{sec:key-observation}, small messages not only have lower latency but also exhibit greater stability, 
further minimizing the risk of correctness violations.

\newcolumntype{P}[1]{>{\centering\arraybackslash}p{#1}}

\begin{table*}[h]
\small
\centering
\ra{1.3}
\begin{tabular} { @{}cccccc@{} }
  \toprule
      &	\textbf{System model} & \textbf{Resilience} & \textbf{Good Case Latency} & \textbf{Pipelining} & \textbf{Optimistic Responsiveness}\\
      \midrule
  \textbf{Sync HotStuff \cite{rot-sync-hotstuff}}	& synchronous  & $f < n/2$ &  $\delta_L$+$\delta_S$+2$\Delta$ & yes & yes \\ 
  \textbf{Tendermint \cite{bucham18}}			& partially synchronous & $f < n/3$  & $\delta_L$+2$\delta_S$ & no & yes \\ 
  \textbf{HotStuff-2 \cite{hotstuff-2}}			& partially synchronous & $f < n/3$  &  3$\delta_L$+2$\delta_S$ & yes & yes \\ 
  \textbf{AlterBFT (this paper)}			& hybrid & $f < n/2$  &  $\delta_L$+$\delta_S$ + 2$\Delta_S$ & yes & yes\\ 
  \textbf{FastAlterBFT (this paper)}			& hybrid & $f < n/2$  &  $\delta_L$ + $\delta_S$ & yes & yes\\
  \bottomrule
  \end{tabular}
  \caption{Protocols in our evaluation and their main characteristics. $\delta_L$ is the actual delay of large messages (i.e., blocks); $\delta_S$ is the actual delay of small messages (i.e., votes and certificates); $\Delta$ is the conservative message delay that accounts for large and small messages; and $\Delta_S$ is the conservative delay of small messages. }
\label{tbl:competition}
\end{table*}

\section{Evaluation}
\label{sec:evaluation}

We compare \alterbft to state-of-the-art leader-rotating Byzantine
consensus protocols in the synchronous and partially synchronous
system models (see Table~\ref{tbl:competition}).  We consider
protocols that allow optimistic responsiveness, meaning that in good
cases, where we have a sequence of honest leaders and synchronous
bounds hold (i.e., after GST), protocols change leaders responsively,
waiting for real network delays only.  
In the synchronous model, we consider a version of Sync HotStuff that
supports responsive leader rotation, pipelining (i.e., replicas start working on the next
block after receiving the certificate for the previous block), and has optimal
latency \cite{rot-sync-hotstuff}.  In the partially synchronous model,
we choose Tendermint \cite{bucham18} and HotStuff-2 \cite{hotstuff-2}
with pipelining, the most recent protocol of the HotStuff
family\cite{hotstuff}.  
Table \ref{tbl:competition}, compares the good case latency of the considered protocols in
failure-free executions.  

\subsection{Experimental Environment and Setup}
\label{sec:eval-setup}

We conducted our experiments on Amazon EC2 with replicas evenly
 distributed across 5 AWS regions: North Virginia (us-east-1), S\~{a}o
 Paulo (sa-east-1), Stockholm (eu-north-1), Singapore
 (ap-southeast-1), and Sydney (ap-southeast-2). A cross-region setup
 within the same provider is a common configuration for performance
 evaluation of BFT consensus protocols designed for
 blockchains~\cite{xft,narwhal,bullshark-async,bullshark-partially-sync,nsdi-DispersedLedger}.
 Replicas were hosted on \emph{t3.medium} instances, with 2 virtual
 CPUs, 4GB of RAM, and running Amazon Linux 2.

To ensure a fair comparison, we implemented all protocols in Go. The
implementations use SHA256 for hashing and Ed25519 64-byte digital
signatures.  We rely on libp2p~\cite{libp2p} for
communication between pairs of replicas.

Each replica includes a built-in client that pre-generates
transactions and stores them in a local pool. When a replica is the
leader of an epoch, it selects transactions from the pool and forms a
block, where the block size determines the number of transactions
included. This design abstracts the mempool (i.e., the component
responsible for propagating client transactions across the system)
from discussion, as different systems may implement it
differently. Therefore, the latencies reported in this paper represent
consensus latencies (i.e., the time required by the leader of an epoch
to commit a block). Throughput is calculated by all replicas as the
rate of committed blocks per time unit. Each point in the graphs is an
average of 3 runs, with each experiment running for 1 minute.

\subsection{On Message Size}
\label{sec:eval-msg-size}

The hybrid synchronous model differentiates between two types of
messages: \typeS and \typeL.  Based on our experimental evaluation
(see Appendix \ref{sec:appendix-msg-delays}), we classified messages
of size 4 KB and lower as type \typeS; larger messages are type \typeL.

Table \ref{tbl:msg-sizes} presents the sizes of all messages exchanged
in \alterbft. The $\vote$ and $\silence$ messages are small,
fixed-size messages that belong to type \typeS. In contrast, the
$\propose$ and $\quit$ messages have variable sizes.

\begin{table}[!ht]
    \small
    \centering
\centering
\ra{1.3}
\begin{tabular} { @{}cP{0.5\columnwidth}c@{} }
  \toprule
\textbf{Message}	& \textbf{Message size (payload)} 				& \textbf{Message type} \\
\midrule
\propose	 			& variable size, depends on block size	& \typeL \\
\vote		 			& fixed size, below 120 bytes 			& \typeS \\
\silence		 		& fixed size, below 100 bytes 			& \typeS \\ 
\quit		& fixed size (50 bytes) + quorum-size * 66 bytes & \typeS \\
\bottomrule
\end{tabular}
\caption{Message sizes in the \alterbft prototype.}
    \label{tbl:msg-sizes}
\end{table}

$\quit$ messages, with certificates, must be exchanged in a
timely manner for correctness and are thus type \typeS.  The size of a certificate
depends on a majority quorum of replicas. In our experiments, the
largest certificate in a system with 85 replicas is 2.8 KB.
Consequently, with the current prototype, we can accommodate
deployments with up to 120 replicas. For larger systems, a more
optimized signature techniques such as BLS \cite{bls} would be
necessary.

The size of $\propose$ messages depends on the block and
certificate sizes. In \alterbft, these messages are
classified as type \typeL, which means there are no restrictions on
their size, and consequently, no restrictions on the block size as
well.

\subsection{Performance Evaluation}
\label{sec:eval-perf}

We measure latency and
throughput while varying the system size (i.e., 25 and 85 replicas)
and block size (i.e., from 1 KB up to 1 MB).

\begin{figure*}[ht!]
    \centering
    \begin{subfigure}[b]{\textwidth}
        \centering
    \includegraphics[width=0.48\textwidth]{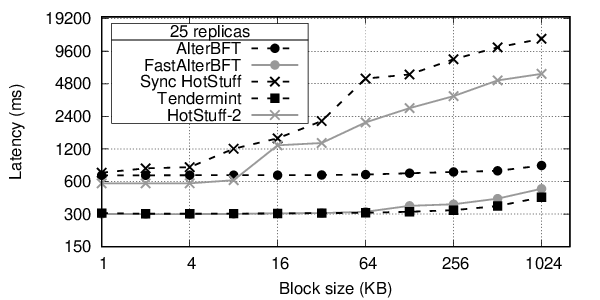}
    \includegraphics[width=0.48\textwidth]{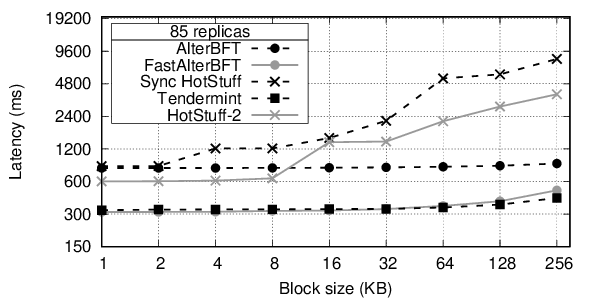}
    \end{subfigure}
    \begin{subfigure}[b]{\textwidth}
        \centering
    \includegraphics[width=0.48\textwidth]{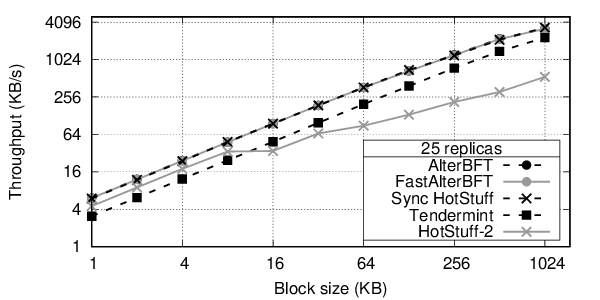}
    \includegraphics[width=0.48\textwidth]{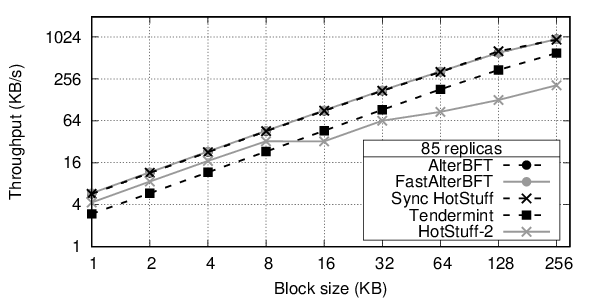}
    \end{subfigure}
    \caption{Average latency (top) and throughput (bottom) comparison for all protocols when varying system size (i.e., 25 and 85 replicas) and block size (all graphs in log scale).
    }
    \label{fig:perf}
\end{figure*}

\PAR{Latency.} 
Figure \ref{fig:perf} shows the average consensus latency computed by
leaders.  From Table \ref{tbl:competition}, Sync HotStuff
and \alterbft latencies directly depend on conservative synchronous
bounds.  This is because both protocols wait for a timeout (i.e.,
$timeoutCommit$, computed as twice the time bound) before committing a
value.  The values used as synchronous bounds for Sync HotStuff
and \alterbft can be found in Appendix \ref{appendix:bounds}.

Synchronous protocols, such as Sync HotStuff, must adopt $\Delta$ that
accounts for the timely delivery of all messages (i.e., large and
small). Consequently, since the size of messages carrying blocks
increases with the block size, the $\Delta$ also increases, leading to
Sync HotStuff's higher latency.  On the other side, \alterbft's
$timeoutCommit$ relies on $\Delta_S$.
As a result, the
difference between \alterbft's and Sync HotStuff's latencies increases
with the block size. Upp to 4 KB blocks, \alterbft
performs slightly better. With 8 KB blocks, \alterbft's
latency is more than 1.5$\times$ lower than Sync HotStuff's, and this
difference raises to 14.9$\times$ with 1 MB blocks.

Latencies of \fastalter and partially synchronous protocols rely only
on the actual message network delays.  Note here that \fastalter
requires one voting phase where it needs to receive the votes from all
replicas, while Tendermint and HotStuff-2 use two voting phases where
they need to receive votes from more than $2/3$ of replicas.
Consequently, \fastalter optimization works only in optimistic
conditions when there are no failures.

Figure~\ref{fig:perf} shows that Tendermint's latency is lower, around
2$\times$, than \alterbft's in all setups.  However, the difference
does not increase with the block size. This is because Tendermint's
and \alterbft's latencies are only affected by one actual network
delay for large messages.  In addition, even though
HotStuff-2's latency depends only on real communication delays,
HotStuff-2 achieves lower latency ($\approx$ 20\%) than \alterbft up
to 8 KB blocks.  The reason is that the actual network delay,
$\delta_L$, increases as the block size increases. Since the latency
of the pipelined version of HotStuff-2 requires three such delays (see
Table \ref{tbl:competition}), the overall latency grows.  As a result,
HotStuff-2 achieves latency from 1.7$\times$ to 7$\times$ higher
than \alterbft when the block size is greater than 8 KB.

Lastly, in failure-free cases, \fastalter's latency is 1.6 to
2.5$\times$ lower than the latency of \alterbft and almost identical
to Tendermint's latency.  This result suggests that, in our wide-area
setup, a voting phase where the leader needs to receive votes from all
replicas requires a similar amount of time as two voting phases with
two-third majority quorums.

\textbf{Throughput.} 
\label{sec:throughput-eval}
Sync HotStuff, \alterbft and \fastalter have similar throughout (see
Figure \ref{fig:perf}), as they have the same communication pattern in
the failure-free case.  Namely, all protocols start working on the
next block as soon as they receive the certificate for the previous
block.  All three protocols outperform partially synchronous protocols
for all system and block sizes, reaching throughput from 1.4$\times$
to 2$\times$ higher than Tendermint's. The reason is that
Tendermint does not use pipelining.  Moreover, they also perform
better, from 1.3$\times$ to 7.2$\times$, than HotStuff-2, a partially
synchronous protocol with pipelining.  Even though HotStuff-2 uses
pipelining, it performs better (1.4$\times$) than
Tendermint only with block sizes up to 8 KB.  With bigger
blocks, HotStuff-2's throughput decreases, becoming worse than
Tendermint's.  This is because as we increase the
block size, the real network delay of messages carrying blocks
increases and varies more. Since HotStuff-2 uses a
linear communication pattern, replicas can receive the proposal only
from the leader. Thus, the overall throughput
decreases as we increase the block size.  Lastly, 
up to 8 KB blocks, the throughput of \alterbft is around
1.4$\times$ better than HotStuff-2. This is because even though both
protocols start ordering the next block after collecting a certificate
for the previous block, the certificate in HotStuff-2 requires votes
from more a two-third majority of replicas, while in \alterbft the
votes from the majority are enough.

\subsection{Additional Results}
\label{sec:extra-results}

Due to space constraints, we now briefly report on additional findings.
More details are presented in the Appendix.

We evaluate the performance
of \alterbft and \fastalter under equivocation attacks, where a
Byzantine leader proposes conflicting blocks to different sets of
replicas (Appendix~\ref{sec:perf-attack}).  These experiments demonstrate the importance of chaining,
even under attack, and highlight that the additional $2\Delta_S$ wait after silence or equivocation certificate
that was required for \fastalter commit rule proves
beneficial for throughput during these attacks.

We explore two alternative
approaches for handling large messages in synchronous consensus
protocols (Appendix~\ref{sec:alternatives}).  The first approach limits consensus instances to small
messages, requiring multiple instances to process larger blocks.  This
approach results in significantly higher latency and lower throughput
due to the overhead of executing additional consensus instances.  The
second approach investigates the effects of sending a large message as
multiple smaller messages, but the experiments showed no improvement
in communication delays.  These findings suggest that combining small
and large messages, as done in \alterbft, provides a more efficient
solution for consensus.

We consider more detailed
data on communication delays across different geographical
regions (Appendix~\ref{sec:appendix-msg-delays}). This data further validates our key observation about small
versus large messages (Section \ref{sec:key-observation}) and helps
establish the bounds used in our performance evaluation
(Section \ref{sec:eval-perf}).

\section {Related Work}
\label{sec:related-work}

In this section, we survey consensus algorithms in different system models. 
Primarily, we focus on protocols that tolerate Byzantine failures and are designed for a blockchain context. 

\PAR{Asynchronous Protocols.}
HoneyBadgerBFT \cite{honey-badger} is the first practical purely
asynchronous consensus protocol designed for blockchain.
It improves on the asynchronous atomic broadcast protocol presented in
\cite{cachin-1}, but relies on expensive $n$ concurrent asynchronous binary
Byzantine agreement (ABBA).  Later protocols
proposed various improvements, such as replacing concurrent ABBA
instances by a single asynchronous multi-value validated Byzantine
agreement (MVBA) \cite{dumbo,s-dumbo}, decoupling transaction
dissemination and agreement \cite{nsdi-DispersedLedger}, and executing them completely concurrently \cite{dumbo-ng}.
Asynchronous protocols are robust but perform worse than partially
synchronous and synchronous protocols.  As a consequence, some
protocols use a simpler leader-based deterministic protocol to improve
the latency in good cases \cite{jolteon-ditto, combo-1, combo-2,
  bolt-dumbo}.

\PAR{Partially Synchronous Protocols.} 
The first practical BFT consensus protocol designed for a partially
synchronous system model is PBFT \cite{pbft}, a leader-based protocol
that can commit a value in three communication steps.  Tendermint
\cite{bucham18} has a failure-free communication pattern similar to
PBFT's, but it is based on a rotating leader.  HotStuff is another
partially synchronous protocol designed for blockchains.  HotStuff 's
main goal is to design a leader rotation mechanism that requires
linear communication $O(n)$ and is responsive, meaning that a new
leader needs to wait just for $n-f$ messages before proposing a value
and not for maximum network delay.  The protocol achieves
responsiveness at the expense of additional communication.  HotStuff-2
\cite{hotstuff-2} shows that HotStuff's additional communication is
not justified in practice and achieves responsiveness with no extra
communication in optimistic conditions, e.g., when we have a sequence
of honest leaders and a synchronous network.

\PAR{Synchronous Protocols.} 
Synchronous BFT consensus protocols require a majority of honest
replicas \cite{fitzi}, as opposed to partially synchronous protocols,
which require a two-third majority.  The first synchronous consensus
designed for blockchains is Dfinity \cite{dfinity}.  Contrary to the
early BFT protocols in the synchronous model
\cite{byz-gen,dolev-strong}, Dfinity does not assume lock-step
execution where replicas execute the protocol in rounds and messages
sent at the start of the round arrive by the end of the
round. Instead, it assumes that replicas start the protocol within
$\Delta$ time.  Dfinity's throughput is affected by the maximum
network delay $\Delta$ because every replica at the beginning of each
round waits for $2\Delta$ before casting a vote.  Abraham et
al. \cite{sync-hotstuff} introduced Sync HotStuff, which removes the
effect of maximum network delay on throughput, achieving throughput
comparable to the partially synchronous HotStuff, and also reducing
latency.  A rotating-leader version of Sync HotStuff was introduced in
\cite{rot-sync-hotstuff}.  \alterbft and rotating-leader Sync HotStuff
share similar common-case behavior. However, they have different epoch
synchronization mechanisms and \alterbft's safety does not require
timely delivery of all messages.

An alternative to the synchronous system model is the ``weak
synchronous model'' \cite{mobile-sluggish}.  The model tolerates
Byzantine replicas and allows some honest replicas to be slow, that
is, communication between slow replicas can violate synchrony bounds.
However, this is true only if the actual number of Byzantine failures
is smaller than $f$.  The first BFT consensus protocol presented in
the weak synchronous model was PiLi \cite{pili}, with latency between
$40\Delta$ and $65\Delta$.  In \cite{sync-hotstuff}, the authors
showed how Sync HotStuff can be adapted to the weak synchronous model.

\PAR{Protocols Based on Extended Hardware.}
Some protocols increase resilience by relying on trusted components.
The main idea is to execute key functionality, such as appending to a
log \cite{a2m} or incrementing a counter \cite{trinc}, inside a
trusted execution environment (e.g., Intel SGX enclaves \cite{sgx}).
Extended hardware has been used to allow both PBFT
\cite{srds-verissimo, a2m,trinc, min-bft} and HotStuff
\cite{hot-stuff-m, damysus} to tolerate a minority of Byzantine
replicas.  Recently, the authors in \cite{dissect-bft} identified
fundamental problems with the deployment of such systems and provided
a solution that requires a two-third majority of honest replicas.
\alterbft does not require any trusted components and relies on
synchrony instead.

Another approach is to divide the system into two parts \cite{tbc}: a
synchronous subsystem that transmits control messages, and an
asynchronous subsystem that transmits the payload.  This model was
generalized to the wormhole hybrid distributed system model where
secure and timely components co-exist \cite{tbc-1, tbc-2}.  \alterbft
also differentiates between two types of messages, but does not assume
the existence of any separate subsystem or special components.

\PAR{DAG-Based Protocols.}
HashGraph \cite{hashgraph} introduced the idea of building a directed
acyclic graph (DAG) of messages and designing an algorithm that will
solve BFT consensus just by interpreting the DAG without sending any
additional messages.  Aleph \cite{aleph} improved the DAG structure by
adding rounds, and a round version of the DAG was efficiently
implemented in Narhwal \cite{narwhal}.  Different versions of
DAG-based BFT consensus protocols that built on Narhwal's DAG have
been proposed for both asynchronous
\cite{narwhal,bullshark-async,dag-rider} and partially synchronous
system models \cite{bullshark-partially-sync, bbca-ledger}. All these
systems tolerate fewer than $1/3$ of Byzantine replicas.  Designing a
synchronous DAG-based protocol that can tolerate a minority of
Byzantine replicas is still an open question.

\PAR{Additional Proposals.}
Thunderella \cite{thunderella} points out that the latency of
synchronous BFT consensus protocol does not need to depend on $\Delta$
when the actual number of faults is less than $1/4$ of the replicas.
Protocols whose latency does not depend on $\Delta$ in some special
conditions are called \emph{optimistically responsive}.  Another
optimistically responsive protocol is XPaxos \cite{xft}, which
achieves optimistic responsiveness by finding a group of $f+1$ honest
and synchronous replicas. XPaxos is only practical when the number of
actual faults is a small constant.  While these protocols are stable
leader protocols, \alterbft is a rotating leader protocol that
achieves responsive latency in the absence of failures in the system
\cite{zyzzyva,sbft}.

The hybrid fault model introduced in \cite{hybrid} distinguishes
between different types of failures and proposes different thresholds
for crash and Byzantine failures. Its most recent refinement
\cite{vft} expands the work by adding the threshold for slow
replicas. This approach allowed the design of more cost-efficient
(tolerating the same number of failures with fewer replicas) protocols
in the data center environment.

An orthogonal approach to our work one could use to boost system
throughput is to separate value propagation from consensus.  Namely,
values are reliably broadcast to replicas using large messages and
consensus, using small messages, is used to establish the order of
value hashes \cite{zarko,narwhal,nsdi-DispersedLedger}.  Consequently,
in the synchronous system model, the $\Delta$ of reliable broadcast
would need to account for the delays of large messages, while
consensus can use a smaller bound that accounts for the delays of
small messages only.  This approach generally hurts latency since it
adds additional communication steps before a replica can commit a
value.

\section{Conclusion}
\label{sec:conclusion}

In this paper, we have introduced the hybrid synchronous system model
and \alterbft, a new Byzantine fault-tolerant hybrid synchronous
consensus protocol.  The hybrid synchronous system model distinguishes
between timing assumptions for small messages, which respect time
bounds, and large messages, which may violate bounds but are
eventually timely.  \alterbft delivers higher throughput with
comparable latency to partially synchronous protocols, while needing
only a $\frac{1}{2}$ majority. It also reduces latency by up to
15$\times$ compared to existing synchronous protocols, with similar
throughput.

\bibliographystyle{plain}
\bibliography{main}

\appendix
\clearpage

\section{Appendix}
\label{sec:appendix}

\subsection{\alterbft's Pseudo-code}
\label{sec:pseudo-code}

Algorithms \ref{alg:alter} and \ref{alg:alter-2} present \alterbft's pseudo-code for normal and abnormal case operations, respectively.

\begin{algorithm*}[htp!]
	\footnotesize
	\begin{algorithmic}[1] 
		\INIT{} 
		\STATE $e_p := 0$  
		\COMMENT{the current epoch} 
		\STATE $hasVoted := false$
		\COMMENT{has the replica voted in the current epoch?} 
		\STATE $lockedBC_p := \nil$ 
		\COMMENT{the most recent block certificate the replica is aware of} 
		\STATE $epochsState_p[] := \nil$
		\COMMENT{an epoch can be in one of the states: \ready, \committed, \notcommitted} 
		\STATE $epochDecision_p[] := \nil$ 
		\COMMENT{an epoch decision can be an $\id$ of a committed block or $\nil$} 
		\ENDINIT 

		\SHORTSPACE 
		\STATE \textbf{when} \emph{bootstrapping} \textbf{do} $StartEpoch(0)$ 
		\COMMENT{the execution starts in epoch 0} 

		\SHORTSPACE 
		\PROCEDURE[upon starting an epoch...]{$StartEpoch(epoch)$} \label{line:tab:startEpoch} 
		\STATE $e_p \assign epoch;$ $hasVoted_p \assign false$ 
		\COMMENT{the replica sets the current epoch, resets $hasVoted$ variable, and...}
		\STATE $epochsState_p[e_p] \assign \ready;$ $epochsDecision_p[e_p] \assign \nil;$ 
		\COMMENT{sets epoch state and epoch decision to $\ready$ and $\nil$, respectively}
		\IF[if the replica is the leader in the current epoch, and...]{$\leader(e_p) = p$}
			\IF[$e_p$ is the first epoch or the replica's $lockedBC_p$ is from the previous epoch...]{$e_p=0$ \textbf{or} $lockedBC_p.epoch = e_p-1$}
			\STATE $Propose()$
			\COMMENT{the replica proposes immediately}
			\ELSE[otherwise,...]
			\STATE \textbf{execute} $Propose()$ \textbf{when} $\timeoutEpochChange(e_p)$ expires
			\label{line:tab:timeout-epoch-change-start-alter}
			\COMMENT{the replica waits for the $2\Delta_S$ timeout to learn the most recent certified block} 
			\ENDIF
		\ENDIF
		\ENDPROCEDURE

		\SHORTSPACE 
		\PROCEDURE[in order to propose...]{$Propose()$} \label{line:tab:propose} 
		\STATE $b \assign getBlock()$
		\COMMENT{the leader gets a new block, and...}
		\STATE \Broadcast\ $\li{\propose,e_p, b, lockedBC_p}$
		\label{line:tab:alter-brd-propose}
		\COMMENT{broadcasts the proposal carrying the block and replica's $lockedBC$}
		\STATE \Broadcast \ $\li{\vote,e_p,id(b)}_p$
		\COMMENT{then, it broadcasts a signed vote, and...}
		\label{line:tab:alter-brd-vote}
		\STATE $hasVoted_p \assign true$ 
		\COMMENT{sets $hasVoted_p$ to avoid voting when it receives a proposal from itself}
		\ENDPROCEDURE

		\SHORTSPACE 
		\FUNCTION[the leader proposes...]{$getBlock()$} \label{line:tab:getBlock} 
		\IF[the block extending the most recent certified block...]{$lockedBC_p \neq \nil$}
		\STATE \textbf{return} $Block\{payload: getPayload(), prev: lockedBC_p.id\}$	
		\COMMENT{the leader knows about, or...}
		\ELSE[the block with no predecessor block if...]
		\STATE \textbf{return} $Block\{payload: getPayload(), prev: \nil\}$	
		\label{tab:newValue}
		\COMMENT{it is not aware of any certified block}
		\ENDIF
		\ENDFUNCTION

		\SHORTSPACE
		\STATE \textbf{when} receive $\li{\propose,e,b,BC}$ \textbf{and} $\li{\vote,e,id(vb)}_c$
		\label{line:tab:alter-proposal-received}
		\COMMENT{when the replica receives a proposal and the vote for it...}
		\INDENT[signed by the leader of the current epoch before detecting any misbehavior...]{\textbf{where} $e=e_p$ \textbf{and} $c=\leader(e)$ \textbf{and} $epochsState_p[e_p] = \ready$ \textbf{do}}
		\label{line:tab:recvProposal-1}
			\STATE \textbf{if} $valid(b) \wedge hasVoted_p = false$ $\wedge$
			\label{line:tab:alter-valid-check}
			\COMMENT{if the block is valid, the replica hasn't voted in the current epoch, and...}
			\INDENT[one of the conditions is fullfilled...]{$(condition_1 \vee condition_2)$  \textbf{then}}
			\label{line:tab:accept-proposal-1} 
			\STATE \Broadcast \ $\li{\vote,e_p,id(b)}_p$
			\COMMENT{the replica broadcast a \vote \ message containing block id, and...}
			\STATE $hasVoted_p \assign true$ 
			\COMMENT{sets $hasVoted_p$ so it does not vote twice, if it receives a forwarded or different proposal}
			\STATE \Forward \ $\li{\vote,e,id(b)}_c$
			\COMMENT{then, it (a) forwards the leader's vote, needed for timely equivocation detection, and...}
			\label{line:tab:forward-vote}	
			\STATE \Forward \ $\li{\propose,e,b,BC}$
			\COMMENT{(b) forwards the received proposal, needed for eventual delivery of all certified blocks}
			\label{line:tab:forward-proposal}
			\ENDINDENT
		\ENDINDENT

		\SHORTSPACE
		\STATE $condition_1 \equiv (lockedBC_p = \nil)$ 
		\COMMENT{the replica is unaware of any certified block}
		\SHORTSPACE
		\STATE $condition_2 \equiv (lockedBC_p \neq \nil \wedge BC \neq \nil \wedge BC.epoch \geq lockedBC_p.epoch)$
		\COMMENT{$BC$ from proposal is more recent than $lockedBC_p$}
		\label{line:tab:condition-2}


		\SHORTSPACE
		\COMMENT{\textbf{***Block Certificate***}}
		\WHEN[upon receiving a block certificate...]{receive $f+1$ distinct $\li{\vote,e,\id(b)}_*$ \textbf{or} $\li{\quit,cert}$ \textbf{where} $cert.type = \valueCert$}
		\label{line:tab:block-cert-received-alter}
		\STATE \textbf{if} {$\li{\quit,cert}$ received} \textbf{then} $c \assign cert$
			\label{line:tab:block-cert-start}
			\COMMENT{it can receive it in a \quit \ message, or...}
			\STATE \textbf{else} $c \assign$ $NewCert$ from  $f+1 \li{\vote,e,\id(b)}_*$ 
			\COMMENT{through $f+1$ individual \vote \ messages}

			\IF[if the certificate is from the current epoch...]{$c.epoch=e_p$}
			\STATE $lockedBC_p \assign c$
				\label{line:tab:update-lockedBC}
				\COMMENT{the replica locks on it by updating its $lockedBC_p$}
			\IF[then, if the replica has not received any other certificate yet...]{$epochsState[e_p] = \ready$}
			\STATE \textbf{start} $\timeoutCommit(e_p, c.id)$
			\label{line:tab:start-timeout-commit-alter}
			\COMMENT{the replica starts $timeoutCommit$}
			\ENDIF
			\STATE   \Broadcast \ $\li{\quit,c}$
			\label{line:tab:alter-block-cert}
			\COMMENT{lastly, the replica broadcasts the certificate,and... }
			\STATE $StartEpoch(e_p+1)$ 
			\label{line:tab:alter-start-epoch-1}
			\COMMENT{starts the next epoch}
			\ELSE[in case the certificate is not from the current epoch...]
			\label{line:tab:epoch-change-update-start}
			\STATE \textbf{if} $\leader(e_p) = p$ $\wedge$
			\COMMENT{if the replica is current epoch leader, and...}
			\INDENT[the certificate is more recent than replica's $lockedBC$...]{($lockedBC_p = \nil \vee c.epoch > lockedBC_p.epoch)$ \textbf{then}}
			\STATE $lockedBC_p \assign c$
				\COMMENT{the replica updates its $lockedBC_p$, and...}
				\STATE   \Broadcast \ $\li{\quit,c}$
			\label{line:tab:block-cert-leader}
			\COMMENT{broadcasts the new certificate}
				\label{line:tab:update-lockedBC-leader}
			\ENDINDENT
			\ENDIF
			
		\ENDWHEN
		\SHORTSPACE
		\COMMENT{\textbf{***Regular Commit Rule***}}
		\WHEN[when $timeoutCommit$ expires...]{$timeoutCommit(e, \id)$ expires} \label{line:tab:onTimeoutCommit} 
		\label{line:tab:commit-rule-1} 
			\IF[if the replica did not observe any proof of misbehavior...]{$epochsState[e] = \ready$} 
			\STATE $epochsState[e] \assign \committed$
			\label{line:tab:alter-commit-1}
			\COMMENT{the replica sets epoch state to $\committed$, and...}
			\STATE $epochsDecision[e] \assign \id$
			\COMMENT{the epoch decision value to $\id$}
			\label{line:tab:commit-rule-1-end}	
			\ENDIF	 
		\ENDWHEN
		\SHORTSPACE
		\COMMENT{\textbf{***Fast Commit Rule (FastAlterBFT)***}}
		\WHEN[when the replica receives votes from all replicas for the proposed block...]{receive $\li{\vote,e,id(b)}_*$ from all replicas}
		\label{line:tab:commit-rule-2} 
			\IF[if the replica did not observe any proof of misbehavior...]{$epochsState[e] = \ready$} 
			\STATE $epochsState[e] \assign \committed$
			\label{line:tab:alter-commit-2}
			\COMMENT{the replica sets epoch state to $\committed$, and...}
			\STATE $epochsDecision[e] \assign id(b)$
			\COMMENT{the epoch decision value to $id(b)$}
			\label{line:tab:commit-rule-2-end}	
			\ENDIF	 
		\ENDWHEN

		\SHORTSPACE 
		\STATE  \textbf{when} receive $\li{\propose,e,b,*}$
		\label{line:tab:alter-receive-alue}
		\COMMENT {when the replica receives a proposal...}
		\INDENT[for a block corresponding to epoch's committed block...]{\textbf{where} $epochDecision[e] = id(b)$ \textbf{do}}
			\label{line:tab:recvProposal}
			\STATE $CommitBlock(b)$
			\label{line:tab:alter-commit-value}
			\COMMENT{the replica commits block $b$ and all its uncommitted predecessor blocks}
		\ENDINDENT
	\end{algorithmic} \caption{\alterbft consensus algorithm: normal case}
	\label{alg:alter} 
\end{algorithm*}

\begin{algorithm*}[htp!]
	\footnotesize
	\begin{algorithmic}[1] 
		\UPON[when a replica enters a new epoch...]{starting the epoch $e$}
		\STATE \textbf{start} $\timeoutCertificate(e_p)$
		\label{line:tab:timeoutCert-start-alter}
		\COMMENT{it starts the timer $\Delta_L+4\Delta_S$ used to detect asynchrony or a malicious leader}
		\ENDUPON 
		\SHORTSPACE

		\WHEN[when $timeoutCertificate$ expires...]{$timeoutCertificate(e)$ expires} \label{line:tab:onTimeoutCertificate} 
		\IF[if the replica did not receive any certificate...]{$e = e_p \wedge epochsState[e] = \ready$} 
			\STATE \Broadcast \ $\li{\silence,e_p}_p$ 
			\label{line:tab:blame-leader}
			\COMMENT{the replica broadcasts a $\silence$ message}
		\ENDIF	
		\ENDWHEN
		\SHORTSPACE
		\COMMENT{\textbf{***Silence Certificate***}} 
		\STATE \textbf{when} receive $f+1$ distinct $\li{\silence,e}_*$ \textbf{or} $\li{\quit,cert}$
		\label{line:tab:alter-blame-cert-start}
		\COMMENT {when a replica receives...}
		\INDENT[a silence certificate...] {\textbf{where} $cert.type = \silCert$ \textbf{do}}
		\STATE \textbf{if} $\li{\quit,cert}$ received \textbf{then} $c \assign cert$
			\COMMENT{from a \quit \ message with the certificate, or...}
			\STATE \textbf{else} $c \assign$ $NewCert$ from  $f+1 \li{\silence,e}_*$  
			\COMMENT{from $f+1$ distinct \silence \ messages}
			\label{line:tab:alter-blame-cert-formed}
			\STATE $MisbehaviorDetected(c)$
			\label{line:tab:misbehavior-1}
			\COMMENT{the replica calls $MisbehaviorDetected$ with the certificate as parameter}
		\ENDINDENT
		\SHORTSPACE
		\COMMENT{\textbf{***Equivocation Certificate***}}  
		\STATE \textbf{when} receive $\li{\vote,e,id(b)}_c$ \textbf{and} $\li{\vote,e,id(v')}_c$ \textbf{or} $\li{\quit,cert}$
		\label{line:tab:alter-equivocation}
		\COMMENT {when a replica receives...}
		\INDENT[an equivocation certificate...]{\textbf{where} $c=\leader(e)$ \textbf{and} $v \neq v'$ \textbf{or} $cert.type = \equivCert$ \textbf{do}}
		\STATE \textbf{if} $\li{\quit,cert}$ received \textbf{then} $c \assign cert$
		\COMMENT{from a \quit \ message with the certificate, or...}
		\STATE \textbf{else} $c \assign$ $NewCert$ from  $\li{\vote,e,id(b)}_c$ \textbf{and} $\li{\vote,e,id(v')}_c$
		\COMMENT{from two distinct \vote \ messages signed by the epoch leader}
		\label{line:tab:alter-equivocation-formed}
		\STATE $MisbehaviorDetected(c)$
		\label{line:tab:misbehavior-2}
		\COMMENT{the replica calls $MisbehaviorDetected$ with the certificate as parameter}
		\ENDINDENT
		
		\SHORTSPACE 
		\PROCEDURE[when $MisbehaviorDetected$ is called...]{$MisbehaviorDetected(cert)$} \label{line:tab:missbehaviorDetected} 
		\IF [if the epoch is still active...]{$epochsState[cert.epoch] = \ready$}
			\STATE $epochsState[cert.epoch] \assign \notcommitted$
			\label{line:tab:not-commit}
			\COMMENT{the replica sets state to \notcommitted}
			\IF[moreover, if $cert$ is the first received certificate for the current epoch...]{$cert.epoch=e_p$}
			\STATE \Broadcast \ $\li{\quit,cert}$
			\label{line:tab:alter-broadcast-misbehavior}
			\COMMENT{the replica broadcasts the certificate, and...}
			\STATE  \textbf{start} $timeoutCommit(e_p, \nil)$ 
			\label{line:tab:silence-cert} 
			\COMMENT{triggers $timeoutCommit(e_p,\nil)$ with a special $\nil$ value} 
		\ENDIF
		\ENDIF
		\ENDPROCEDURE

		\SHORTSPACE
		\WHEN[when $timeoutCommit$ for epoch $e$ with a $\nil$ value expires...]{$timeouCommit(e, \nil)$ expires} \label{line:tab:OntimeoutCommit} 
		\label{line:tab:timeout-epoch-change-start}
		\IF[if the replica is still in epoch $e$...]{$e = e_p$} 
			\STATE $StartEpoch(e_p+1)$ 
			\COMMENT{the replica starts the next epoch}
			\label{line:tab:start-epoch-2}
			
		\ENDIF	
		\ENDWHEN
	\end{algorithmic} \caption{\alterbft consensus algorithm: handling malicious leaders and asynchrony}
	\label{alg:alter-2} 
\end{algorithm*}

\subsection{Proof of Correctness}
\label{sec:appendix-proof}
In this section, we present \alterbft's proof of correctness.
The proof consists of five parts: epoch advancement, safety, liveness, block availability, and external validity. 

\subsubsection{Epoch Advancement}
The epoch advancement mechanism ensures all honest replicas move through epochs continuously and start each epoch within 
$\Delta_S$ time. It assumes all honest replicas start epoch 0 within $\Delta_S$ time. 
Notice here that this mechanism only relies on the timely delivery of type \typeS messages. 

\begin{lemma}
	\label{lemma:alter-epoch-start}
	Every honest replica always progresses to the next epoch.
  \end{lemma}
  \begin{proof}
	Assume, for the sake of contradiction, that there exists an honest replica $r$ that remains in some epoch $e$ indefinitely. 
	This would imply that in epoch $e$, $r$ did not generate any of the certificates $C_e(B_k)$, $C_e(\silence)$, or $C_e(\double)$. 
  
	However, upon entering epoch $e$, every honest replica starts the $timeoutCertificate(e)$ timer (line \ref{line:tab:timeoutCert-start-alter} in Algorithm \ref{alg:alter-2}). 
	When this timeout expires, if an honest replica has not received any certificate, it broadcasts the \silence\ message (lines \ref{line:tab:onTimeoutCertificate}--\ref{line:tab:blame-leader} 
	in Algorithm \ref{alg:alter-2}). 
  
	Therefore, if no certificate is formed before the $timeoutCertificate(e)$ expires, all honest replicas will broadcast the \silence\ message, 
	resulting in the formation of the blame certificate $C_e(\silence)$. This contradicts the assumption that an honest replica can stay in epoch 
	$e$ indefinitely. Thus, every honest replica must progress to the next epoch.
  \end{proof}


  \begin{lemma}
	\label{lemma:alter-epoch-advancement}
	If an honest replica starts epoch $e$ at time $t$, then all honest replicas will start epoch $e$ by time $t + \Delta_S$.
  \end{lemma}
  \begin{proof}
	Suppose an honest replica $r$ starts epoch $e$ at time $t$. This implies that $r$ either received and broadcasted $C_{e-1}(B_k)$ at time 
	$t$ (line \ref{line:tab:alter-block-cert} in Algorithm \ref{alg:alter}), or received and broadcasted $C_{e-1}(\silence)$ or $C_{e-1}(\double)$ 
	at time $t - timeoutCommit(2\Delta_S)$ (line \ref{line:tab:alter-broadcast-misbehavior} in Algorithm \ref{alg:alter-2}).
  
	Since messages with certificates (\quit\ messages) are of type \typeS, they will be delivered within $\Delta_S$ time. Therefore, in the first case, all honest replicas receive $C_{e-1}(B_k)$ by time $t + \Delta_S$ and start epoch $e$. In the second case, all honest replicas receive $C_{e-1}(\silence)$ or $C_{e-1}(\double)$ by time $t - \Delta_S$ and subsequently start epoch $e$ within $2\Delta_S$, resulting in the same deadline of $t + \Delta_S$.
  
	It is also possible that while an honest replica is waiting for $timeoutCommit(e-1)$ to expire, it may receive $C_{e-1}(B_k)$. In such a case, 
	the replica will broadcast $C_{e-1}(B_k)$, and start epoch $e$ (lines \ref{line:tab:alter-block-cert}--\ref{line:tab:alter-start-epoch-1} in Algorithm \ref{alg:alter}). 
	All honest replicas will then receive this message and, if they have not already done so, will start epoch $e$.
  
	Therefore, all honest replicas start epoch $e$ by time $t + \Delta_S$.
  \end{proof}

  \begin{theorem}(Epoch synchronization)
	\label{theorem:alter-continuous-epochs}
	All honest replicas continuously move through epochs, with each replica starting a new epoch within $\Delta_S$ time of any other honest replica.
	\end{theorem}
	
	\begin{proof}
	We prove this theorem by combining Lemma \ref{lemma:alter-epoch-start} and Lemma \ref{lemma:alter-epoch-advancement}.
	
	First, from Lemma \ref{lemma:alter-epoch-start}, we know that every honest replica always moves to the next epoch. 
	This ensures that no honest replica remains stuck in any epoch indefinitely.
	
	Second, from Lemma \ref{lemma:alter-epoch-advancement}, we know that if an honest replica starts epoch $e$ at time $t$, 
	then all honest replicas start epoch $e$ by time $t+\Delta_S$. This guarantees that all honest replicas start each epoch 
	within $\Delta_S$ time of each other.
	
	Combining these two results, we can conclude that all honest replicas continuously move through epochs, with each replica 
	initiating a new epoch within $\Delta_S$ time of any other honest replica.
	\end{proof}

\subsubsection{Safety}
The following Lemmas and Theorem are related to the \alterbft's safety. Namely, the protocol ensures all replicas 
agree on the same blockchain (i.e., forks does not exist). 

\begin{lemma} 
	\label{lemma:alter-main}
	If an honest replica directly commits block $B_k$ in epoch $e$, then 
	(i) no block different from $B_k$ can be certified in epoch $e$, and 
	(ii) every honest replica locks on block $B_k$ in epoch $e$.  
  \end{lemma}
  \begin{proof}
	\alterbft has two commit rules. We need to show that the lemma holds in both scenarios. 
  
	First, consider the general case where an honest replica $r$ directly commits $B_k$ at time $t$ because $timeoutCommit(e) = 2\Delta_S$ 
	expired and it did not receive any blame or equivocation certificate (lines \ref{line:tab:commit-rule-1}--\ref{line:tab:commit-rule-1-end} 
	in Algorithm \ref{alg:alter}). This implies that at time $t-2\Delta_S$, $r$ received $C_e(B_k)$, locked on it, and started $timeoutCommit(e)$. 
	Additionally, $r$ broadcast $C_e(B_k)$. Since this message is of type \typeS, all honest replicas received $C_e(B_k)$ within $\Delta_S$ time, 
	by $t-\Delta_S$.
  
	For part (i), assume for contradiction that an honest replica $q$ received and voted for a block $B_l \neq B_k$ in epoch $e$. Since every honest 
	replica votes only once, $q$ must have received the proposal and leader's vote for $B_l$ before receiving $C_e(B_k)$, i.e., at time 
	$t_1 < t-\Delta_S$. Upon voting for $B_l$, $q$ broadcast the leader's vote (line \ref{line:tab:forward-vote} in Algorithm \ref{alg:alter}). 
	Consequently, $r$ would receive the leader's vote for $B_l$ by $t_1 + \Delta_S$, which is before $t$. Moreover, since $r$ received $C_e(B_k)$ 
	at $t-2\Delta_S$, we know that at least one honest replica voted for $B_k$ at some moment $t_2 < t-2\Delta_S$. Therefore, $r$ would receive the 
	leader's vote for $B_k$ by $t_2 + \Delta_S$. Since both leader's votes for $B_k$ and $B_l$ would arrive at $r$ before $t$, a $C_e(\double)$ 
	certificate would be constructed, and $r$ would not commit (line \ref{line:tab:not-commit} in Algorithm \ref{alg:alter-2}). 
	This is a contradiction. Therefore, property (i) holds as no honest replica votes for a block different from $B_k$, otherwise $r$ would not commit.
  
	For part (ii), it suffices to prove that every honest replica receives $C_e(B_k)$ before moving to the next epoch. This is sufficient because, 
	due to (i), $B_k$ is the only certified block in epoch $e$, and since $e$ is the current epoch, there is no more recent block certificate. 
	Consequently, if an honest replica receives $C_e(B_k)$ in epoch $e$, it will update its $lockedBC$ to it (line \ref{line:tab:update-lockedBC} 
	in Algorithm \ref{alg:alter}). Since we know all honest replicas will receive $C_e(B_k)$ by $t-\Delta_S$, we need to prove that no honest 
	replica will start epoch $e+1$ before $t-\Delta_S$. 
  
	Assume, for contradiction, that an honest replica $q$ moves to epoch $e+1$ at $t_1 < t-\Delta_S$ without receiving $C_e(B_k)$. Since $C_e(B_k)$ 
	is the only block certificate in epoch $e$, $q$ must have moved to epoch $e+1$ because it received $C_e(\silence)$ or $C_e(\double)$. Since $q$ 
	broadcasts $C_e(\silence)$ or $C_e(\double)$ (line \ref{line:tab:alter-broadcast-misbehavior} in Algorithm \ref{alg:alter-2}) at time $t_1$, 
	$r$ would receive them by $t_1 + \Delta_S$. Since $t > t_1 + \Delta_S$, $r$ would not commit $B_k$, a contradiction. Note that waiting for 
	$timeoutCommit(e) = 2\Delta_S$ (line \ref{line:tab:silence-cert} in Algorithm \ref{alg:alter-2}) after receiving $C_e(\silence)$ or $C_e(\double)$ 
	is not needed in this case.
  
	Now consider the case where $r$ commits due to the \fastalter commit rule (lines \ref{line:tab:commit-rule-2}--\ref{line:tab:commit-rule-2-end} 
	in Algorithm \ref{alg:alter}). Specifically, this means $r$ starts $timeoutCommit(e)$ at $t-2\Delta_S$ and commits at some moment $t_1 < t$ after 
	receiving votes from all replicas. Part (i) trivially holds because if $r$ received votes for $B_k$ from all replicas, this means that all honest 
	replicas (f+1) voted for $B_k$. Since honest replicas vote only once in an epoch, no other $B_k' \neq B_k$ can collect (f+1) votes and be certified
	 in epoch $e$. 
  
	For part (ii), every replica needs to wait $timeoutCommit(e) = 2\Delta_S$ before moving to the next epoch in case it receives $C_e(\silence)$ or 
	$C_e(\double)$ first (line \ref{line:tab:silence-cert} in Algorithm \ref{alg:alter-2}). Assume, for contradiction, that an honest replica $q$ 
	moved to epoch $e+1$ before receiving $C_e(B_k)$, namely before $t-\Delta_S$. Again, due to (i), replica $q$ moved to epoch $e+1$ because it 
	received $C_e(\silence)$ or $C_e(\double)$. Due to the extra $2\Delta_S$ timeout, it must have received one of these certificates at some moment 
	$t_2 < t-\Delta_S-2\Delta_S$. Since $q$ would forward the received certificate at $t_2$, all honest replicas, including $r$, would receive this 
	certificate by $t_2 + \Delta_S$, and since this is before $t-2\Delta_S$, $r$ would not start $timeoutCommit(e)$ at $t-2\Delta_S$ and would not 
	commit, a contradiction.   
  
	Therefore, both parts (i) and (ii) hold.
  \end{proof}

  \begin{lemma}
	\label{lemma:alter-all-lock}
	If $C_e(B_k)$ is the only certified block in epoch $e$ and $f + 1$ honest replicas lock on block $B_k$ in epoch $e$, then in all epochs $e' > e$ these replicas will only vote for blocks that extend $B_k$.
  \end{lemma}
  \begin{proof}
  Let set $C$ contain $f+1$ or more honest replicas that lock on $B_k$ in epoch $e$.
  We prove this lemma by induction on the epoch number.
  
  \textbf{Base step} ($e' = e + 1$) :
  A replica $r \in C$ will only vote for a block $B_{k'}$ in epoch $e'$ if $B_{k'}$ extends a block certified in an epoch greater than or equal to $e$ (line \ref{line:tab:condition-2} in Algorithm \ref{alg:alter}). 
  Since $e$ is the previous epoch and the highest in the system, and $B_k$ is the only certified block in epoch $e$, the lemma holds trivially for $e' = e+1$.
  
  \textbf{Induction step} ($e' \rightarrow e' + 1$):
  Assume the lemma holds for until epoch $e' + 1$. We will show it holds for $e' + 1$ also.
  
  From the induction hypothesis, in epochs $e + 1$ to $e' + 1$, replicas in $C$ only vote for blocks that extend $B_k$. Let $B_l$ be the last block to receive $f + 1$ vote messages in some epoch $e''$ where $e+1 \leq e'' \leq e'-1$. 
  Therefore, for all replicas in $C$, $lockedBC = C_{e''}(B_l)$ and it follows that $B_l$ extends $B_k$. As a result, a replica will only vote for a block $B_{k'}$ in $e'$ if $B_{k'}$ extends $B_l$ and therefore $B_k$.
  
  By induction, the lemma holds for all epochs $e' > e$.
  \end{proof}

  \begin{lemma}
	\label{lemma:alter-uniq-ext}
	If an honest replica directly commits block $B_k$ in epoch $e$, then any
	block $B_l$ that is certified in epoch $e' > e$ must extend $B_k$. 
	\end{lemma}
	\begin{proof}
	The proof follows directly from Lemmas \ref{lemma:alter-main} and 
	\ref{lemma:alter-all-lock}. More precisely, if an honest replica directly
	commits block $B_k$ in epoch $e$, by Lemma \ref{lemma:alter-main}, we know that 
	$f+1$ honest replicas (set $C$) lock on block $B_k$ in epoch $e$ and $B_k$ is the only certified block in epoch $e$. Consequently,
	by Lemma \ref{lemma:alter-all-lock}, replicas from $C$ vote only for the blocks
	extending block $B_k$ in epochs $e'>e$. Therefore, no block $B_l$ that 
	does not extend $B_k$ can collect $f+1$ votes and thus cannot be 
	certified in any epoch $e'>e$. 
	\end{proof}
  
  \begin{theorem}(Safety)
  \label{theorem:alter-safety}
  No two honest replicas commit different blocks at the same height. 
  \end{theorem}
  \begin{proof}
  Suppose, for the sake of contradiction, that two distinct blocks $B_k$ and $B_k'$ 
  are committed for the height $k$. Suppose $B_k$ is committed as a result of $B_l$ 
  being directly committed in epoch $e$ and $B_k'$ is committed as a result 
  of $B_{l'}$ being directly committed in epoch $e'$. 
  Without loss of generality, assume $l<l'$. Note that all directly 
  committed blocks are certified. This is true because both commit rules require that replica receives $C_e(B_k)$ before directly commiting $B_k$ in epoch $e$ (lines \ref{line:tab:commit-rule-1} and \ref{line:tab:commit-rule-2} in Algorithm \ref{alg:alter}). 
  By Lemma \ref{lemma:alter-uniq-ext}, $B_{l'}$ extends $B_l$. Therefore, $B_k = B_k'$ which is a required contradiction.
  
  \end{proof}

  \subsubsection{Liveness}
  
  The following Lemmas and Theorem are related to the \alterbft's liveness. Namely, the protocol ensures that the new blocks 
  are continuously committed and added to the blockchain. 

  \begin{lemma} 
	\label{lemma:alter-commit-epoch}
	If the epoch $e$ is after GST and the leader of the epoch is an honest replica, all honest replicas commit a block in this epoch.
  \end{lemma} 
  \begin{proof}
	Consider an epoch $e$ with an honest leader $l$, occurring after GST. Let $t > GST$ be the time when the first honest replica starts epoch $e$. 
	By Lemma \ref{lemma:alter-epoch-advancement}, all honest replicas enter epoch $e$ by time $t + \Delta_S$. Consequently, they all broadcast their 
	$lockedBC$ by time $t + \Delta_S$ at the latest. As a result, $l$ will receive certificates from all honest replicas by time $t + 2\Delta_S$. 
	This is why $l$ needs to wait for $timeoutEpochChange(e) = 2\Delta_S$ after entering the epoch if it does not know the certificate from the 
	previous epoch, to update its $lockedBC$ to the most recent certificate (lines \ref{line:tab:timeout-epoch-change-start-alter} and 
	\ref{line:tab:epoch-change-update-start}--\ref{line:tab:block-cert-leader} in Algorithm \ref{alg:alter}). 
  
	Consequently, the honest leader $l$ broadcasts \mymsg{\propose,e,B_k,lockedBC_l}{} and \mymsg{\vote, e, id(B_k)}{l} by time $t + 3\Delta_S$ at 
	the latest. Since we are after GST, all honest replicas receive both messages within $\Delta_L$ time, by time $t + 3\Delta_S + \Delta_L$ and vote 
	for the proposal. The votes are of type \typeS and all honest replicas receive them within $\Delta_S$ time, form a block certificate, and start 
	$timeoutCommit(e)$ by time $t + 4\Delta_S + \Delta_L$. 
  
	Given that the earliest point when an honest replica entered epoch $e$ is $t$ and honest replicas set $timeoutCertificate(e)$ to expire in 
	$4\Delta_S + \Delta_L$, no honest replica will send a \mymsg{\silence,e}{*} message in epoch $e$, and $C_e(\silence)$ cannot be formed. 
	Furthermore, since $l$ is honest, it does not equivocate, so no $C_e(\double)$ can be formed in epoch $e$ either. 
  
	Consequently, when $timeoutCommit(e)$ expires, all honest replicas will commit $B_k$ and all its ancestors.
  \end{proof}

  \begin{theorem}(Liveness) All honest replicas keep committing new blocks.  
  \label{theorem:liveness}
  \end{theorem}
  \begin{proof}
  By the Theorem \ref{theorem:alter-continuous-epochs} replicas move through epochs. Eventually, after $GST$,
  replicas will reach epochs with honest leaders.  
  Consequently, by the Lemma \ref{lemma:alter-commit-epoch} all honest replicas will commit blocks in these epoch. 
  \end{proof}

  \subsubsection{Block Availability}
  \alterbft allows replicas to commit a block $B_k$ before receiving the actual block (line \ref{line:tab:commit-rule-1} and \ref{line:tab:commit-rule-2} in Algorithm 1). In this section we prove 
  that the protocol ensures $B_k$ and all its ancestors blocks will eventually be received by all honest replicas.  
  
  \begin{lemma}
	\label{lemma:certified-ancestors}
	Every block $B_k$ (where $k \neq 0$) proposed by an honest replica in some epoch $e$ has, as its ancestors, blocks that have been certified in one of the epochs $e' < e$.       
  \end{lemma}
  \begin{proof}
	The proof for this lemma directly follows from Algorithm \ref{alg:alter}. Specifically, a leader $l$ of epoch $e$ that proposes block $B_k$, which extends some block $B_l$, must provide a valid block certificate for block $B_l$ from some epoch $e' < e$ (line \ref{line:tab:propose} in Algorithm \ref{alg:alter}). 
  
	Furthermore, an honest replica will only vote for $B_k$ if $B_k.prev = id(B_l)$, a check that is part of the $valid()$ function (line \ref{line:tab:alter-valid-check} in Algorithm \ref{alg:alter}).
  \end{proof}

  \begin{theorem}(Block availability)
	All blocks committed by honest replicas will eventually be received by all honest replicas.
  \end{theorem}
  \begin{proof}
	Assume an honest replica $r$ commits block $B_k$. We know that $B_k$ must be certified before being committed 
	(lines \ref{line:tab:start-timeout-commit-alter}, \ref{line:tab:commit-rule-1}, and \ref{line:tab:commit-rule-2} in Algorithm \ref{alg:alter}). 
	By Lemma \ref{lemma:certified-ancestors}, all of $B_k$'s ancestors are also certified blocks.
  
	For a block to be certified, at least one honest replica must vote for it. Additionally, an honest replica, along with the vote, 
	forwards the proposal (line \ref{line:tab:forward-proposal} in Algorithm \ref{alg:alter}). Consequently, if a block is certified at time $t$, 
	at least one honest replica forwards the proposal before time $t$. 
  
	Since the \propose\ message is of type \typeL, we know, by the communication properties of type \typeL messages (Section \ref{prop-big-msg}),
	 that it will be received by all honest replicas before $\textit{max}\{t, GST\} + \Delta_L$.
  \end{proof}
  
  \subsubsection{External Validity}
\alterbft ensures all committed blocks are valid. 
  \begin{theorem}(External validity)
	Every committed block satisfies the predefined \emph{valid()} predicate.
\end{theorem}
\begin{proof}
	We know that block must be certified before being committed 
	(lines \ref{line:tab:start-timeout-commit-alter}, \ref{line:tab:commit-rule-1}, and \ref{line:tab:commit-rule-2} in Algorithm \ref{alg:alter}). 
	By Lemma \ref{lemma:certified-ancestors}, all of $B_k$'s ancestors are also certified blocks.
	This implies that at least one honest replica accepted these blocks, meaning that $valid()$ returned true for these blocks on at least one honest replica (line \ref{line:tab:alter-valid-check} in Algorithm \ref{alg:alter}). 
\end{proof}
  
  \section{Synchronous Bound $\Delta$}
  \label{appendix:bounds}
  
  This section details how we calculated the synchronous bound $\Delta$ used in our experiments (Section \ref{sec:eval-perf}). The $\Delta$ is based on the 
  delays collected and presented Appendix \ref{sec:appendix-msg-delays}.
  We used 99.99\% percentile latency of collected values as the conservative bound \cite{xft}. Table \ref{tbl:delta} shows 99.99 \% percentile 
  latency for messages of different sizes. We can see that, as we increase the message size, the 99.99\% percentile latency increases. 
  
  Synchronous bound $\Delta$ of classical synchronous protocols needs to account for all messages (large and small). Consequently, we calculated the size of the  
  largest message transmitted in Sync HotStuff protocol in all setups considered: different block and system sizes. Moreover, we adopt $\Delta$ to account 
  for messages of maximum size. Table \ref{tbl:sync-hot-stuff-delta} shows the $\Delta$ and the message size it accounts for, which we used 
  when running Sync HotStuff in our evaluation. We can see that $\Delta$ increases with block size because it needs to account for 
  messages carrying blocks. Specifically, the largest message sent in Sync HotStuff is a proposal message that contains both the block and the certificate \cite{rot-sync-hotstuff}. 
  As a result, in Table \ref{tbl:sync-hot-stuff-delta} for a block size of 1 KB and system size of 25 replicas, the proposal message is of size 2 KB: 
  1 KB (block size) + 1 KB (certificate size).

  \begin{table*}[!ht]
    \footnotesize
    \centering
  \begin{tabular}{|c|c|c|c|c|c|c|c|c|c|c|c|c|} \hline
  \textbf{Message size (KB)}& 1  & 2  & 3  & 4  & 8  & 16  & 32  & 64  & 128  & 256  & 512  & 1024  \\ \hline
  \textbf{99.99 \% (ms)	}		& 254  & 273  & 308 & 325  & 514 & 663 & 995 & 2594 & 2825 & 3935 & 5080 & 6099 \\ \hline
  
  \end{tabular}
  \caption{99.99 \% percentile of collected message delays for different message sizes. 99.99 \% percentile is based on values collected during one day long experiments.}
    \label{tbl:delta}
  \end{table*}
  
  \begin{table*}[h!]
  \footnotesize
  \centering
  \begin{tabular}{|c|c|c|c|c|c|c|c|c|c|c|c|} \hline
  \textbf{Block size (KB)}	& \textbf{1}  & \textbf{2}  & \textbf{4}  & \textbf{8}  & \textbf{16}  & \textbf{32}  & \textbf{64}  & \textbf{128}  & \textbf{256}  & \textbf{512}  & \textbf{1024}   \\ \hline
  \textbf{25 replicas}		& \makecell{273 ms \\ (2 KB)} & \makecell{308 ms \\ (3 KB)} & \makecell{325 ms \\ (5 KB)} & \makecell{514 ms \\ (9 KB)} & \makecell{663 ms \\ (17 KB)} & \makecell{995 ms \\ (33 KB)} & \makecell{2594 ms \\ (65 KB)} & \makecell{2825 ms \\ (129 KB)} & \makecell{3935 ms \\ (257 KB)} & \makecell{5080 ms \\ (513 KB)} & \makecell{6099 ms \\ (1025 KB)} \\ \hline
  \textbf{85 replicas}			& \makecell{325 ms \\ (4 KB)} & \makecell{325 ms \\ (5 KB)} & \makecell{514 ms \\ (7 KB)} & \makecell{514 ms \\ (11 KB)} & \makecell{663 ms \\ (19 KB)} & \makecell{995 ms \\ (35 KB)} & \makecell{2594 ms \\ (67 KB)} & \makecell{2825 ms \\ (131 KB)} & \makecell{3935 ms \\ (259 KB)} & \makecell{5080 ms \\ (515 KB)} & \makecell{6099 ms \\ (1027 KB)} \\ \hline
  
  \end{tabular}
  \caption{Sync HotStuff's synchronous conservative bound $\Delta$ for different block and system sizes. Table's fields show: $\Delta$ (message size it accounts for).}
  \label{tbl:sync-hot-stuff-delta}
  \end{table*}
  
  \begin{table*}[h!]
  \centering
  \footnotesize
  \begin{tabular}{|c|c|c|c|c|c|c|c|c|c|c|c|} \hline
  \textbf{Block size (KB)}& \textbf{1}  & \textbf{2}  & \textbf{4}  & \textbf{8}  & \textbf{16}  & \textbf{32}  & \textbf{64}  & \textbf{128}  & \textbf{256}  & \textbf{512}  & \textbf{1024}   \\ \hline
  \textbf{25 replicas}		& \multicolumn{11}{c|}{254 ms (1 KB)} \\ \hline
  \textbf{85 replicas}		& \multicolumn{11}{c|}{308 ms (3 KB)} \\ \hline
  
  \end{tabular}
  \caption{\alterbft's synchronous conservative bound $\Delta_S$ for different block and system sizes. Table's fields show: $\Delta_S$ (message size it accounts for).}
  \label{tbl:alter-delta}
  \end{table*}
  
  \alterbft differentiates between two synchronous bounds: $\Delta_S$ and $\Delta_L$. $\Delta_S$ needs to account for 
  type \typeS messages and needs to hold always, while $\Delta_L$ accounts for type \typeL messages and needs to hold only eventually. 
  The messages carrying blocks in \alterbft are of type \typeL, and since their timely delivery is needed only for progress, we can use 
  less conservative bounds. In our experiments, we used 99\% percentile latency. Moreover, from Table \ref{tbl:competition}, we can see that 
  $\Delta_L$ does not affect \alterbft's latency. Conversely, latency is affected by $\Delta_S$.
  Since timely delivery of type \typeS messages is needed for safety, we used 99.99\% percentile latency for $\Delta_S$. 
  However, since type \typeS messages are small (up to 3 KB in our setups), the values adopted for $\Delta_S$ are much smaller. 
  Table \ref{tbl:alter-delta} shows the values we used for \alterbft's $\Delta_S$ in our experiments.

\section{Additional Results}
\label{sec:appendix-additional-results}

In this section, we present additional experimental results that, while excluded from the main paper due to space constraints, we believe are important.

\subsection{Performance under Attack}
\label{sec:perf-attack}

We now evaluate \alterbft and \fastalter under equivocation attacks.  
In the equivocation attack, the Byzantine leader of an epoch sends one proposal to half of the replicas and another proposal to the other half; Byzantine replicas vote for both proposals. 
Figure \ref{fig:attacks} presents the data for a system of 
25 replicas with 128 KB blocks while varying the number of Byzantine replicas from 2 to 12.  




\begin{figure}[h]
    \centering
    \includegraphics[width=0.8\columnwidth]{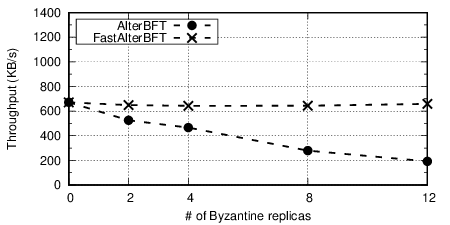}
    \includegraphics[width=0.8\columnwidth]{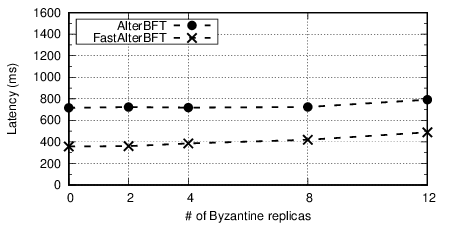}
	\caption{\alterbft and \fastalter throughput (top) and latency (bottom) under equivocation attack, 25 replicas and 128 KB blocks.}
	\label{fig:attacks}
\end{figure}

\fastalter's throughput is not much affected by equivocation attacks due to the chaining mechanism.
Honest replicas will not commit a block in epochs with a Byzantine leader, but 
if they gather a certificate for one of the two blocks proposed, in epochs with an honest leader, the leader will extend and indirectly commit one of these blocks.
Since in \fastalter replicas wait for 2$\Delta_S$ after receiving an equivocation certificate, they always receive a certificate for one of the blocks 
before moving to the next epoch. As a result, the throughput is almost identical to that in the failure-free case. 
In \alterbft, replicas move to the next epoch immediately after receiving the equivocation certificate. 
Consequently, blocks proposed by Byzantine leaders are often wasted. 
As we increase the number of Byzantine replicas, more epochs are wasted, and \alterbft's overall throughput decreases. 

Moreover, the equivocation attack does not have a significant effect on the latency of the protocols. With 2 Byzantine replicas, the latencies stay
the same. As we increase the number of faulty replicas to 4, 8 and 12, the latency of \alterbft increases by 2\%,2\%, and 7.5\%, respectively, while the latency of 
\fastalter increases by 6.5\%, 15.7\%, and 33\%. The increase is because blocks proposed by Byzantine replicas will not be committed 
in the epochs in which they were proposed but in the first epoch with an honest leader. Since in \fastalter we have more such blocks than in \alterbft, the impact on average latency is more significant. 

Lastly, we can see that mandatory 2$\Delta_S$ delay of \fastalter has an overall positive effect on performance when the equivocation attack is in place. 
Together with the benefits presented in failure-free case (see Section \ref{sec:eval-perf}), this result serves as a compelling argument for \fastalter adoption.

\subsection{Design Alternatives}
\label{sec:alternatives}

In this section, we evaluate possible alternatives for the hybrid model and \alterbft. 
We consider two approaches for synchronous consensus protocols that build on the fact that small messages have reduced and more stable communication delays than large messages (see Section \ref{sec:key-observation}): 
\begin{enumerate}
    \item Limiting the size of values that are ordered in an instance of consensus to a few thousand bytes (i.e., small messages). In this case, synchrony bound needs to account for 
    small messages only but multiple consensus instances are needed to order blocks bigger than the chosen value size.
     \item Sending every large message as many small messages. A large block can use a single instance of consensus in this case, but a replica can only act on a large block after it has received all smaller messages that correspond to the original block.
   
\end{enumerate}

We evaluate the first alternative approach described above and compare it to \alterbft and to Sync HotStuff, where large blocks require conservative synchrony bounds. 
More precisely, we measure the throughput and latency of Sync HotStuff, where to order a 128 KB block
the leader uses 64 consensus instances.
In each instance, the leader proposes a 2 KB chunk (Chunked-HS). We compare it to Sync HotStuff, where
a leader uses one consensus instance but sets a conservative synchrony bound (Sync HotStuff).
In these experiments, we use the original Sync HotStuff \cite{sync-hotstuff} with a stable leader since it is
unclear how the technique could be used with a rotating leader (i.e., how would every leader know which block chunk to propose?).
Figure \ref{fig:chunked} shows results for 25 replicas.

\begin{figure}[h]
    \centering
    \includegraphics[width=0.85\columnwidth]{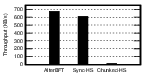}
    \includegraphics[width=0.85\columnwidth]{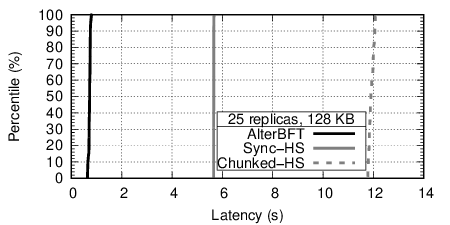}
	\caption{Performance comparison of synchronous consensus 
    with chunked proposals (Chunked-HS), conservative bounds (Sync-HS), and \alterbft for
    25 replicas with 128 KB blocks.}
	\label{fig:chunked}
\end{figure}

Chunked-HS performs worse than Sync HotStuff with a conservative 
bound: It has 2$\times$ higher latency and 55$\times$ lower throughput. 
The reason behind this lies in the overhead of additional consensus executions.
Even though Sync HotStuff starts multiple instances in parallel, it cannot start
the next instance before certifying the proposal of the current instance, so it needs
two communication steps before starting a new instance. 
In conclusion, empirical evidence suggests that consensus protocols are better off combining small and large messages, instead of resorting to small messages only.




\begin{figure}[h]
    \centering
    \includegraphics[width=0.85\columnwidth]{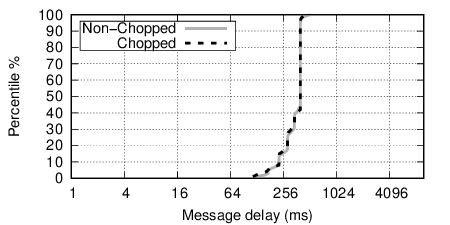}
	\caption{Message delays between N. Virginia and S. Paulo (x-axis in log scale) when sending 128 KB messages (Non-Chopped) versus sending 64 2 KB messages (Chopped).}
	\label{fig:chunked-2}
\end{figure}

To evaluate the second alternative approach described above, we repeated the same experiments used in Section \ref{sec:key-observation} to collect message delays between different AWS regions, but instead of sending one large message,
we divided the messages into small messages and measured the time needed for those small messages to reach their destination and a response to come back (i.e., round-trip time). 
Figure \ref{fig:chunked-2} compares message delays between N. Virginia and S. Paulo when sending one 128 KB message as a whole (Non-Chopped) and sending 64 2 KB messages (Chopped).
We can see that delays observed are almost identical. 
We conclude that chopping large messages into small messages does not reduce communication delays.
Therefore, a large message chopped into small messages is subject to the timeouts of large messages.

\subsection{Message Communication Delays}
\label{sec:appendix-msg-delays}

Figures on pages 24--28 display the message delays between servers deployed across different geographical regions on AWS and DigitalOcean. 
Table \ref{tbl:regions} lists the server numbers, their locations, and the respective providers. In the graphs, the server numbers correspond 
to specific instances based on their location and provider.

The results reveal a consistent trend across all regions, confirming the key observation discussed in Section \ref{sec:key-observation}: 
smaller messages, up to 4 KB, exhibit lower and more stable delays, with the difference becoming more pronounced as message size increases.

  \begin{table}[hb]
    \centering
  \begin{tabular}{|c|c|c|} \hline
  \textbf{Server \#} & \textbf{Location} & \textbf{Provider} \\ \hline
  0 & North Virginia & AWS \\ \hline
  1 & Sao Paulo & AWS \\ \hline
  2 & Stockholm & AWS \\ \hline
  3 & Singapore & AWS \\ \hline
  4 & Sydney & AWS \\ \hline
  5 & New York & DigitalOcean \\ \hline
  6 & Toronto & DigitalOcean \\ \hline
  7 & Frankfurt & DigitalOcean \\ \hline
  8 & Singapore & DigitalOcean \\ \hline
  9 & Sydney & DigitalOcean \\ \hline
  
  \end{tabular}
  \caption{The server numbers, their locations, and their providers. }
    \label{tbl:regions}
  \end{table}
  
  \begin{figure*}[ht!]
    \center
    \includegraphics[width=\textwidth]{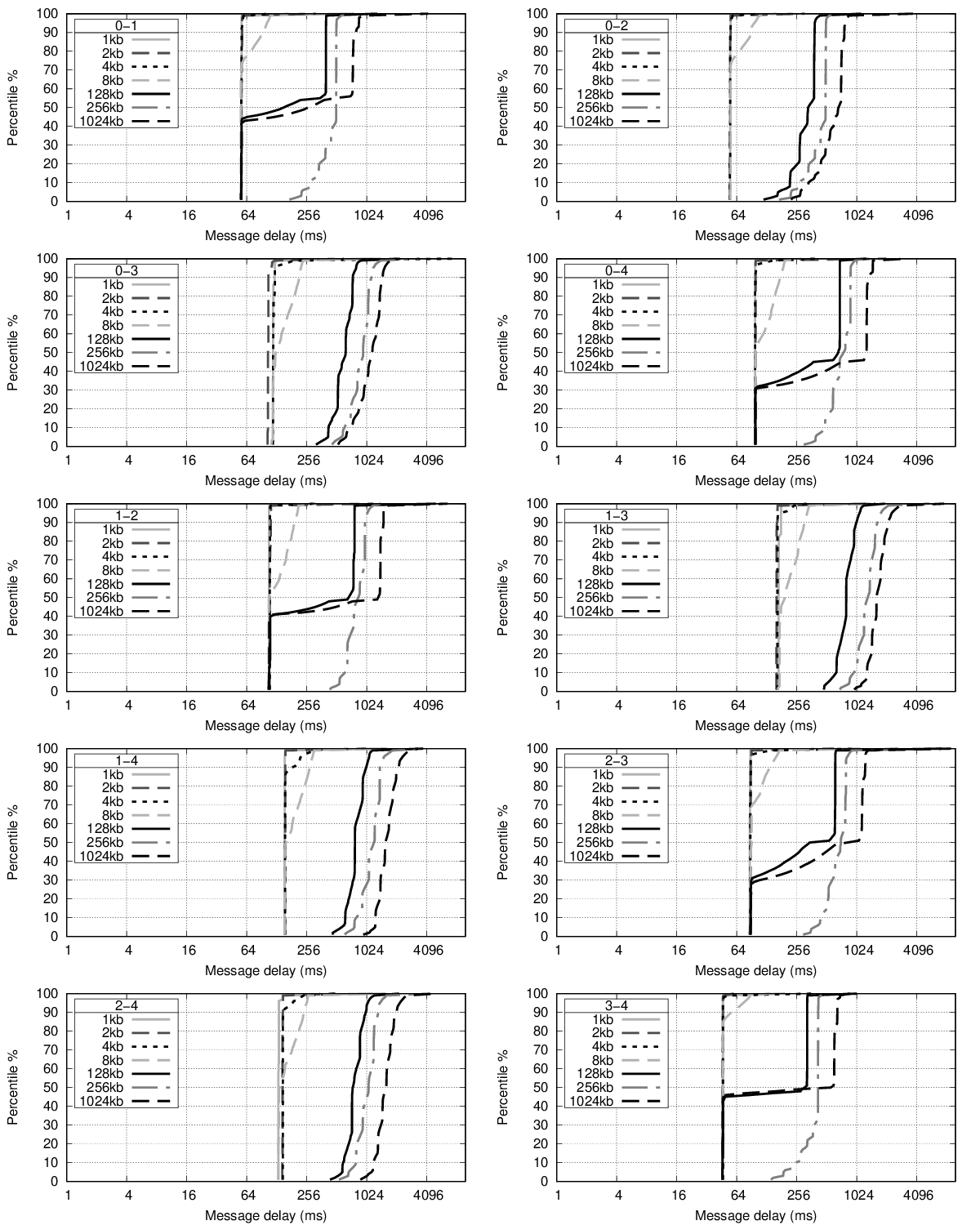}
    \vspace{4in}
    \label{fig:message-delays}
\end{figure*}

\begin{figure*}[ht!]
  \center
  \includegraphics[width=\textwidth]{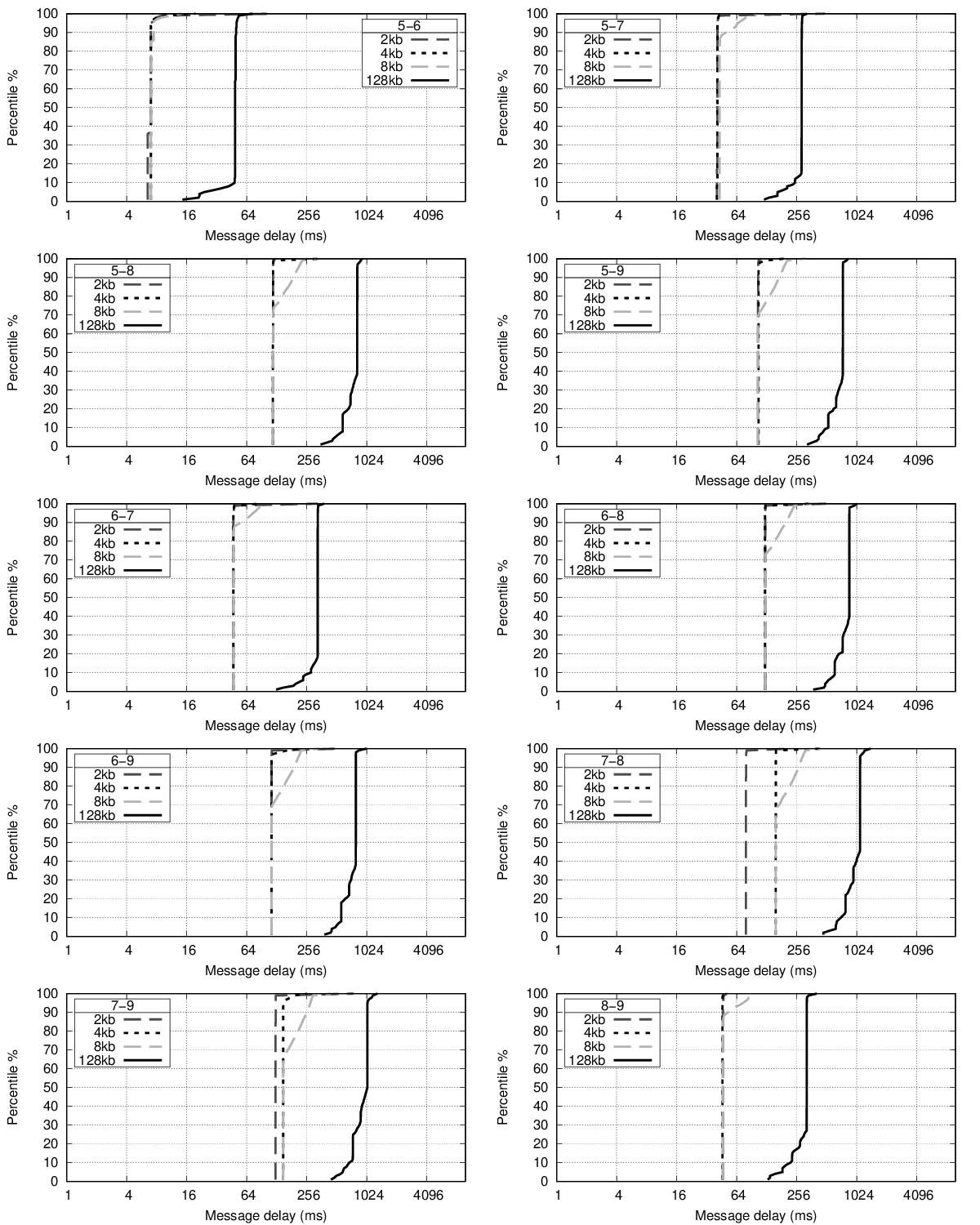}
  \vspace{4in}
  \label{fig:message-delays}
\end{figure*}

\begin{figure*}[ht!]
  \center
  \includegraphics[width=\textwidth]{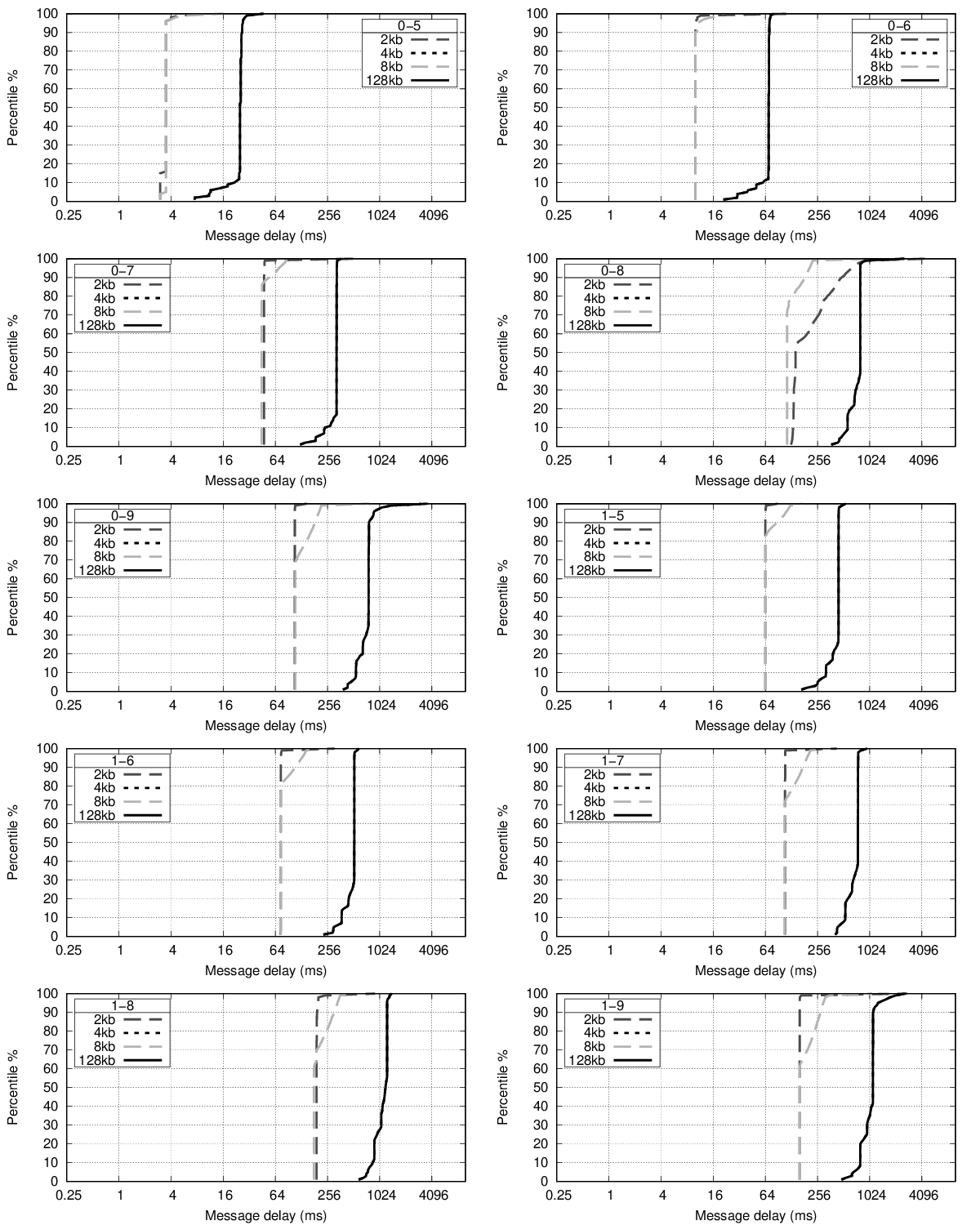}
  \vspace{4in}
  \label{fig:message-delays}
\end{figure*}

\begin{figure*}[ht!]
  \center
  \includegraphics[width=\textwidth]{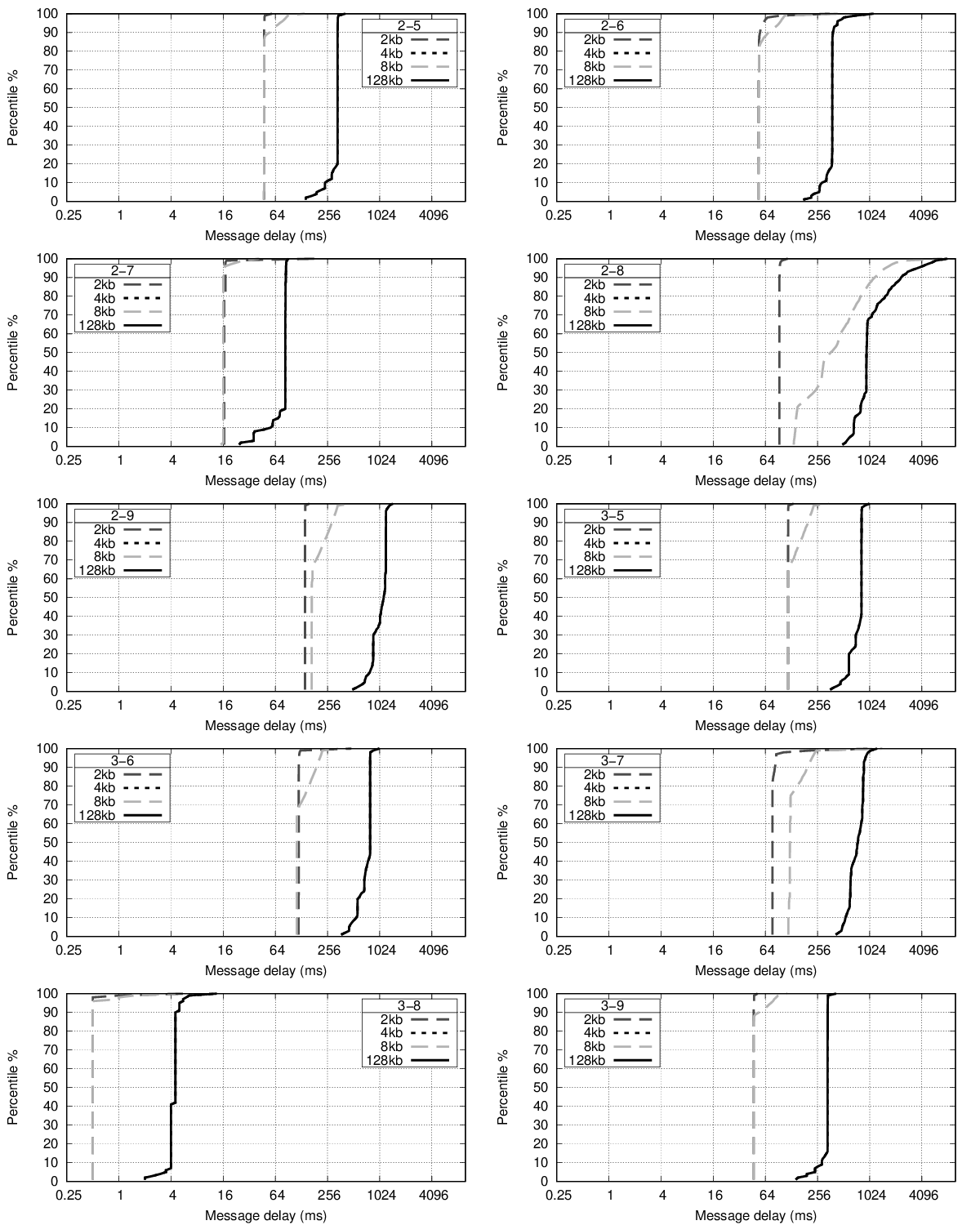}
  \vspace{4in}
  \label{fig:message-delays}
\end{figure*}

\begin{figure*}[ht!]
  \center
  \includegraphics[width=\textwidth]{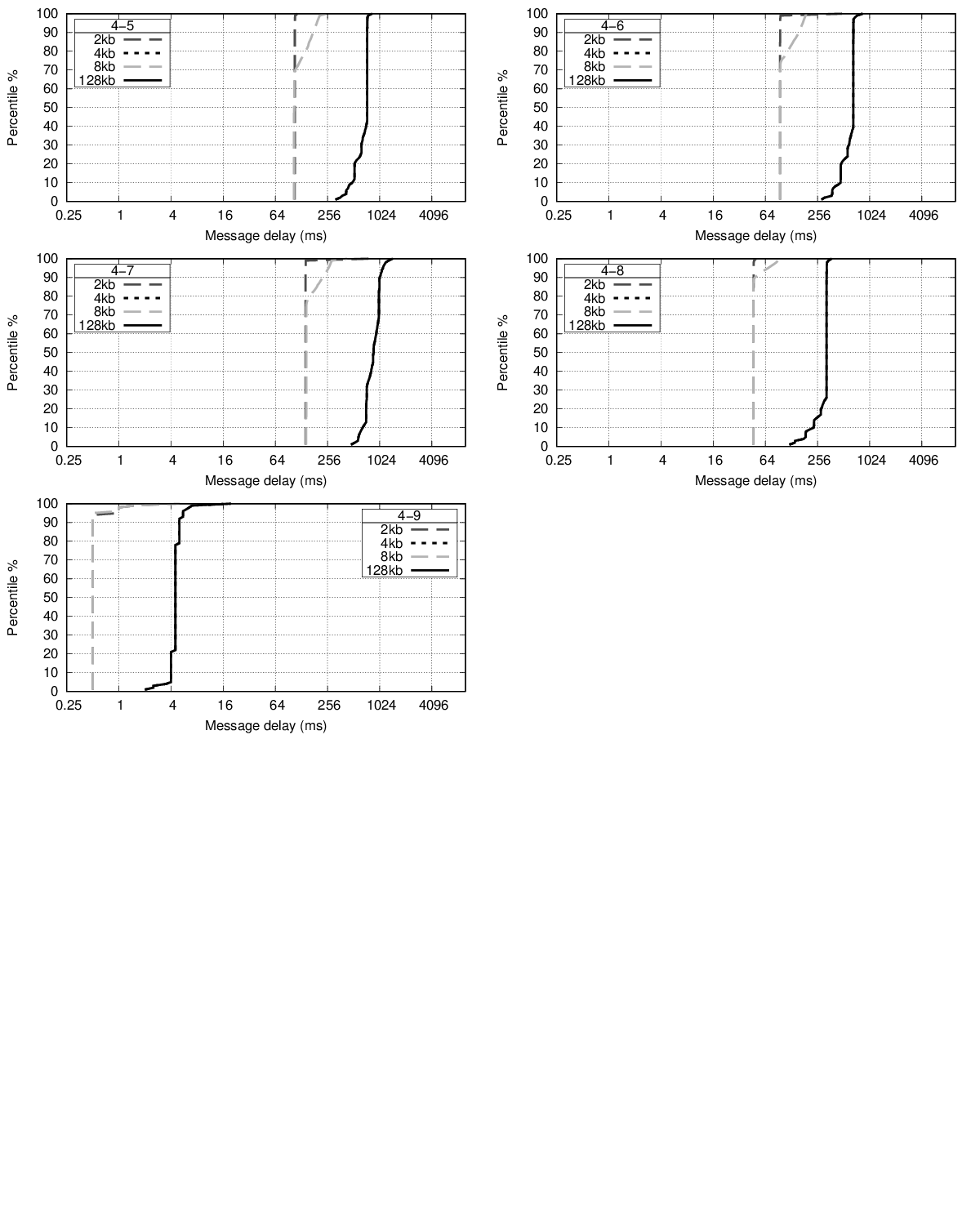}
  \vspace{4in}
  \label{fig:message-delays}
\end{figure*}

\end{document}